\documentclass[twoside,english]{article}
\usepackage[T1]{fontenc}
\usepackage[latin9]{inputenc}
\usepackage{geometry}
\usepackage{color}
\usepackage{verbatim}
\usepackage{mathrsfs}
\usepackage{amsmath}
\usepackage{amsthm}
\usepackage{amssymb}
\usepackage{stmaryrd}
\usepackage{setspace}
\usepackage{tikz}
\usepackage[authoryear]{natbib}
\usepackage{xargs}[2008/03/08]
\makeatletter
\theoremstyle{plain}
\newtheorem{thm}{\protect\theoremname}
\theoremstyle{definition}
\newtheorem{example}{\protect\examplename}
\theoremstyle{plain}
\newtheorem{prop}{\protect\propositionname}
\theoremstyle{plain}
\newtheorem{cor}{Corollary}
\theoremstyle{plain}
\newtheorem{lem}{\protect\lemmaname}

\@ifundefined{date}{}{\date{}}
\usepackage[T1]{fontenc}
\usepackage[latin9]{inputenc}
\usepackage{geometry}
\usepackage{color}
\usepackage{mathrsfs}
\usepackage{mathtools}
\usepackage{amsmath}
\usepackage{amsthm}
\usepackage{graphicx}
\usepackage{amssymb}
\usepackage{setspace}
\usepackage{multirow}
\usepackage[authoryear]{natbib}
\usepackage{xargs}[2008/03/08]
\usepackage{ifpdf} % part of the hyperref bundle
\ifpdf % if pdflatex is used

 \IfFileExists{lmodern.sty}{\usepackage{lmodern}}{}

\fi % end if pdflatex is used

\usepackage{enumitem}
\usepackage{ifthen}
\usepackage{natbib}
\def\LyX{\texorpdfstring{%
  L\kern-.1667em\lower.25em\hbox{Y}\kern-.125emX\@}
  {LyX}}

\newcommand{\mylabel}[2]{#2\def\@currentlabel{#2}\label{#1}}
\newcommand{\vertiii}[1]{{\left\vert\kern-0.25ex\left\vert\kern-0.25ex\left\vert #1
    \right\vert\kern-0.25ex\right\vert\kern-0.25ex\right\vert}}

\usepackage{fancyhdr}
\makeatletter
\newcommand\Author{Lin, M\"uller, Park}
\title{Manifold Additive Models}
\let\Title\@title

\usepackage[hidelinks]{hyperref}
\hypersetup{
    colorlinks,
    linkcolor={red!50!black},
    citecolor={blue!50!black},
    urlcolor={blue!80!black}
}

\newcommand{\footremember}[2]{%
    \footnote{#2}
    \newcounter{#1}
    \setcounter{#1}{\value{footnote}}%
}

\newcounter{condA}
\newcounter{condB}

\makeatother

\usepackage{babel}
\providecommand{\examplename}{Example}
\providecommand{\lemmaname}{Lemma}
\providecommand{\propositionname}{Proposition}
\providecommand{\theoremname}{Theorem}

\begin{document}
\global\long\def\expect{\mathbb{E}}%
\global\long\def\prob{\mathrm{Pr}}%
\global\long\def\bW{\mathbf{W}}%
\global\long\def\bX{\mathbf{X}}%
\global\long\def\bu{\mathbf{u}}%
\global\long\def\bx{\mathbf{x}}%
\global\long\def\proj{\mathcal{P}}%
\global\long\def\real{\mathbb{R}}%
\global\long\def\Op{O_{P}}%
\global\long\def\manifold{\mathcal{M}}%
\global\long\def\txm{T_{x}\mathcal{M}}%
\global\long\def\op{o_{P}}%
\global\long\def\oas{o_{a.s.}}%
\global\long\def\hilbert{\mathcal{\mathcal{L}}^{2}}%
\global\long\def\covarop{\mathcal{C}}%
\global\long\def\asymplt{\lesssim}%
\global\long\def\rkhs{\mathcal{H}(\mathbf{K})}%
\global\long\def\rk{\mathbf{K}}%
\global\long\def\asympgt{\gtrsim}%
\global\long\def\asympeq{\asymp}%
\global\long\def\xsample{\mathcal{X}}%
\global\long\def\ltwospace{\mathcal{L}^{2}(D)}%
\global\long\def\linearopset{\mathfrak{L}}%
\global\long\def\diffop{\mathrm{d}}%
\newcommandx\tangentspace[2][usedefault, addprefix=\global, 1=\manifold]{T_{#2}#1}%
\global\long\def\asympstgt{\gg}%
\global\long\def\vecfield{\mathcal{V}}%
\global\long\def\dim{d}%
\global\long\def\coordop{\mathbf{V}}%
\global\long\def\ball{B}%
\global\long\def\vec#1{\mathbf{#1}}%
\newcommandx\lpnorm[3][usedefault, addprefix=\global, 1=r, 2=]{\|#3\|_{\mathcal{L}^{#1}}^{#2}}%
\newcommandx\lp[1][usedefault, addprefix=\global, 1=p]{\mathcal{L}^{#1}}%
\global\long\def\transpose{\mathrm{\top}}%
\global\long\def\transport{\tau}%
\global\long\def\lielog{\mathfrak{log}}%
\global\long\def\liealgebra{\mathfrak{g}}%
\global\long\def\lieexp{\mathfrak{exp}}%

\global\long\def\asympstlt{\ll}%
\global\long\def\ysample{\mathcal{Y}}%
\global\long\def\ltwo{\mathcal{L}^{2}}%
\global\long\def\observed{\mathscr{D}}%
\global\long\def\innerprod#1#2{\langle#1,#2\rangle}%
\global\long\def\metric#1#2{\langle#1,#2\rangle}%
\global\long\def\Log{\mathrm{Log}}%
\global\long\def\Exp{\mathrm{Exp}}%
\newcommandx\vfnorm[3][usedefault, addprefix=\global, 1=\mu, 2=]{\|#3\|_{#1}^{#2}}%
\newcommandx\vfinnerprod[2][usedefault, addprefix=\global, 1=\mu]{\llangle#2\rrangle_{#1}}%
\global\long\def\define{:=}%
\global\long\def\differentiate{D}%
\global\long\def\tdomain{\mathcal{T}}%
\global\long\def\tm{T\manifold}%
\global\long\def\liegroup{\mathcal{G}}%
\global\long\def\groupop{\oplus}%
\global\long\def\vec#1{\mathbf{#1}}%
\global\long\def\chspace{\mathrm{LT}_{+}(m)}%
\global\long\def\lowtri{\mathrm{LT}(m)}%

\global\long\def\mannorm#1{\vertiii{#1}}%
\newcommandx\opnorm[3][usedefault, addprefix=\global, 1=\mu, 2=]{\vertiii{#3}_{#1}^{#2}}%
\newcommandx\fronorm[2][usedefault, addprefix=\global, 1=]{|#2|_{F}^{#1}}%
\global\long\def\borel{\mathscr{B}}%
\def\lin#1{{#1}}
\def\SPD{$\mathcal{S}^{+}$}
\def\SPDD{$\mathcal{S}^{+}\,\,$}
\global\long\def\spd{\mathcal{S}^{+}_m}

\global\long\def\dtispd{\mathcal{S}_{3}^{+}}%
\global\long\def\sym{\mathcal{S}(m)}

\def\red{\textcolor{red}}

% ***NOTES ****

% Biometrika does not allow abbreviations so SPD is not an option while the 
%SYM notation is long and will deter statistical readers so I changed it but can be 
%easily changed back due to placeholders  see below for old commands 

%Note for now we can have some abbreviations later but not in the first couple of sections 

%Also check the counter for examples, Btka wants each category numbered separately, ie lemmas and 
%theorems, examples etc need to be counted separately. While we do not need to 
%follow all of their style requirements at this time, this one we should implement. 

%\global\long\def\spd{\mathrm{Sym_{\star}^{+}}(m)}

%\global\long\def\spd{\mathrm{Sym_{\star}^{+}}(m)}%
%\global\long\def\dtispd{\mathrm{Sym_{\star}^{+}}(3)}%
%\global\long\def\sym{\mathrm{Sym}(m)}

%\setlength{\labelsep}{20pt}
%\setlength{\itemindent}{200pt}

%\title{Additive Models for  Symmetric Positive-Definite Matrices}
%\title{Manifold Additive Models, with Application to  Symmetric Positive-Definite Matrices}
\title{Additive Models for  Symmetric Positive-Definite Matrices, Riemannian Manifolds and Lie groups}
%\author{Zhenhua Lin, Hans-Georg M\"uller, and Byeong U. Park\\
%National University of Singapore, University of California, Davis, and Seoul National University}

\author{Zhenhua Lin$^1$\footremember{lin}{Corresponding author; email: linz@nus.edu.sg. Research was partially supported by NUS startup grant
		R-155-000-217-133.} \and Hans-Georg M\"uller$^2$\footremember{muller}{Email: hgmueller@ucdavis.edu. Research was supported in part by NSF grant DMS-2014626.} \and Byeong U. Park$^3$\footremember{park}{Email: bupark@stats.snu.ac.kr. The work was supported by Samsung Science and Technology Foundation (project number SSTF-BA1802-01).}}

\date{%
	\textit{$^1$National University of Singapore\\
			$^2$University of California, Davis\\
			$^3$Seoul National University}%
}

\maketitle

\begin{abstract}
In this paper an additive regression model for a symmetric positive-definite matrix valued response and multiple scalar predictors is proposed. The model exploits the abelian group structure inherited from either the Log-Cholesky metric or the Log-Euclidean framework that turns the space of symmetric positive-definite matrices into a Riemannian manifold and further a bi-invariant Lie group. The additive model for responses in the space of symmetric positive-definite matrices with either of these metrics is shown to connect to an additive model on a tangent space. This connection not only entails an efficient algorithm to estimate the component functions but also allows to generalize the proposed additive model to general Riemannian manifolds that might not have a Lie group structure. Optimal asymptotic  convergence rates and normality of the estimated component functions are also established. Numerical studies show that the proposed model enjoys superior numerical performance, especially when there are multiple predictors. The practical merits of the proposed model are demonstrated by analyzing diffusion tensor brain imaging data. 
\end{abstract}


\section{Introduction\label{sec:Introduction}}

Data in the form of symmetric positive-definite matrices arise in many areas, including computer vision
\citep{Caseiro2012,Rathi2007}, signal processing \citep{Arnaudon2013,Hua2017},
medical imaging \citep{Dryden2009,Fillard2007} and neuroscience \citep{Friston2011},
among other fields and applications. For instance, they are used to model brain functional
connectivity that is often characterized by covariance matrices of blood-oxygen-level dependent
signals \citep{Huettel2008}. In diffusion tensor imaging analysis \citep{LeBihan1991},
a $3\times3$ symmetric positive  matrix that is computed  for each voxel describes the dominant shape of local diffusion
of water molecules.

The space  \SPDD  of symmetric positive matrices is a nonlinear metric space and, depending on the metric, forms a 
Riemannian manifold. Various metrics have been studied 
\citep{
pigo:14}; one criterion for the choice of  the metric is to avoid the swelling effect in the geodesics connecting two elements of \SPDD \citep{Arsigny2007} that negatively affects the Frobenius metric and various other metrics. The abundance of \SPD-valued data
in many areas stands in contrast  with the relative  sparsity of work  on their  statistical analysis, in particular regarding regression with \SPD-valued
responses, which is the theme of this paper.  Existing work includes Riemannian frameworks to analyze
diffusion tensor images with a focus on averages and modes of variation
\citep{Fletcher2007,penn:06} and various versions of %, statistical inference under a parametric
%model \citep{Schwartzman2006}, and 
nonparametric regression such
as spline regression \citep{Barmpoutis2007}, local constant regression
\citep{Davis2010}, intrinsic local linear regression \citep{Yuan2012},
wavelet regression \citep{Chau2019a} and Fr\'echet regression \citep{pete:19}.  Various metric, manifold and Lie
group structures have been proposed, for example, the trace metric
\citep{Lang1999}, affine-invariant metric (also called Fisher--Rao
metric) \citep{Moakher2005,penn:06,Fletcher2007}, Log-Euclidean metric
\citep{Arsigny2007}, Log-Cholesky metric \citep{Lin2019a}, scaling-rotation
distance \citep{Jung2015} and Procrustes distance \citep{Dryden2009}. As the \SPD~manifold is a Riemannian manifold and more generally a metric space, 
regression techniques developed for general Riemannian manifolds \citep[e.g., ][among many others]{Pelletier2006,Shi2009, Steinke2010a,Davis2010,Fletcher2013,Hinkle2014,Cornea2017} and metric spaces \citep{Hein2009,Petersen2019,Lin2019b} also apply to the \SPD~space.

%In this paper, we study additive regression for a response variable taking values
%in the space of SPD matrices. 

Additive regression originated with \cite{ston:85:1} and is known to be an efficient
way of avoiding the well known curse of dimensionality problem that one faces in nonparametric regression
when the dimension of the covariate vector increases but so far has been by and large limited to the case of real-valued and functional responses. Examples for additive regression approaches for  real-valued responses include
the original work on smooth backfitting \citep{Mammen1999}, its extensions to
generalized additive models \citep{Yu2008}, to additive quantile models \citep{Lee2010},
to generalized varying coefficient models \citep{Lee2012}, and to the case of
errors-in-variables \citep{Han2018}. Additive models for functional responses include 
additive functional regression based on spline basis representation \citep{Scheipl2015},
smooth backfitting via real-valued singular components \citep{Park2018},  and  modeling 
 density-valued responses \citep{Han2020} with transformations \citep{pete:16}.  Recently, a general framework
for Hilbert-space-valued responses has been developed \citep{Jeon2020}.

This paper contains  three major contributions.  First, to the best of our knowledge, this is the first paper to study additive
regression for \SPD-valued responses. As theoretically and numerically demonstrated below,  additive regression  is less prone to the curse of dimensionality while maintaining a high degree of flexibility in the spirit of structured  nonparametric modeling. In contrast, previous studies for modeling  \SPD-valued responses focused on ``full'' nonparametric regression such as local constant/polynomial regression that are well
known to be subject  to the curse of dimensionality when there are many predictors. Second, by focusing on the Log-Cholesky
and Log-Euclidean frameworks that endow the space \SPDD
with an abelian Lie group structure and a bi-invariant metric, we  propose a novel 
intrinsic group additive regression model that exploits the abelian group structure of the  manifold \SPDD  in a regression setting for the first time. This sets our work apart, as previously only the general manifold structure of \SPDD was considered  in  regression approaches.  Third, we show that this group additive model  can be transformed into an additive model
on tangent spaces by utilizing the Riemannian logarithmic map. This not only leads to an efficient way to estimate the additive component functions, but also paves the way for extending  the additive model 
to other more general manifolds, leading to a general approach to manifold additive modeling. 

\section{Methodology\label{sec:Methodology}}

\subsection{Preliminaries on Manifolds}\label{subsec:manifold}

The proposed approaches for manifold additive modeling  are closely tied to the manifold structure of the response space in a general regression model, where we showcase the proposed approaches for the space \SPDD of symmetric positive-definite matrices. To properly define the proposed manifold additive models we  require some  basic notions for Riemannian manifolds and Lie groups that are compiled in the following.  
Let $\manifold$ be a simply connected and smooth manifold modeled on a $D$-dimensional Euclidean space. %Let $g$ be a metric tensor that turns $\manifold$ into a geodesically complete Riemannian manifold. 
The tangent space $\tangentspace[\manifold]y$
at $y\in\manifold$ is a linear space consisting of velocity vectors
$\alpha^{\prime}(0)$ where $\alpha:(-1,1)\rightarrow\manifold$ represents
a differentiable curve passing through $y$, i.e., $\alpha(0)=y$.
Each tangent space $\tangentspace[\manifold]y$ is endowed with an
inner product $g_{y}$ that varies smoothly with $y$ and thus is a $D$-dimensional Hilbert space with the induced norm denoted
by $\|\cdot\|_{y}$. The inner products $\{g_{y}:y\in\manifold\}$
are collectively denoted by $g$, referred to as the Riemannian metric
of $\manifold$ that also defines a distance $d$ on $\manifold$.

A geodesic $\gamma$ is a curve defined on {[}$0,\infty)$ such that
for each $t\in[0,\infty)$, $\gamma([t,t+\epsilon])$ is the shortest
path connecting $\gamma(t)$ and $\gamma(t+\epsilon)$ for all sufficiently
small $\epsilon>0$. The Riemannian exponential map $\Exp_{y}$ at
$y\in\manifold$ is a function mapping $\tangentspace[\manifold]y$
into $\manifold$ and defined by $\Exp_{y}(u)=\gamma(1)$ with $\gamma(0)=y$
and $\gamma^{\prime}(0)=u\in\tangentspace[\manifold]y$. Conversely,
$\gamma_{y,u}(t)=\Exp_{y}(tu)$ is a geodesic starting at $y$ and
with direction $u$. For a tangent vector $u\in\tangentspace[\manifold]y$, the cut time $c_u$ is the positive number such that 
$\gamma_{y,u}([0,c_u])$
is a shortest path connecting $\gamma_{y,u}(0)$ and $\gamma_{y,u}(c_u)$,
but $\gamma_{y,u}([0,c_u+\epsilon])$ is not a shortest path for any $\epsilon>0$. Let  $\mathcal{E}_{y}=\{\Exp_y(tu):u\in\tangentspace[\manifold]y,\|u\|_y=1,0\leq t<c_u\}$. The inverse of $\Exp_{y}$, denoted by $\Log_{y}$ and called the Riemannian logarithmic map at $y$, can be defined by $\Log_{y}z=u$ for $z\in\mathcal{E}_{y}$ such that $\Exp_{y}u=z$.

%The cut locus of $y$ in the tangent space consists of elements $u\in\tangentspace[\manifold]y$ such that $\gamma_{y,u}([0,1])$
%is a shortest path connecting $\gamma_{y,u}(0)$ and $\gamma_{y,u}(1)$,
%but $\gamma_{y,u}([0,1+\epsilon])$ is not for any $\epsilon>0$. The cut locus of $y$ in $\manifold$, denoted
%by $\mathcal{C}_{y}$, is the image of the cut locus of $y$ in the
%tangent space under $\Exp_{y}$. The inverse of $\Exp_{y}$, denoted
%by $\Log_{y}$ and called the Riemannian logarithmic map at $y$,
%can be defined by $\Log_{y}z=u$ for $z\in\manifold\backslash\mathcal{C}_{y}$
%such that $\Exp_{y}u=z$.

A vector field $U$ is a function defined on $\manifold$ such that
$U(y)\in\tangentspace[\manifold]y$. The Levi--Civita covariant derivative
$\nabla$ on $\manifold$ is a torsion-free bilinear form that, at
$y\in\manifold$, maps a tangent vector $v\in\tangentspace[\manifold]y$
and a vector field $U$ to another tangent vector $\nabla_{v}U\in\tangentspace[\manifold]y$.
Given a curve $\gamma(t)$ on $\manifold$, $t\in I$ for a real interval $I$, a vector field $U$ along $\gamma$ is a smooth map defined
on $I$ such that $U(t)\in\tangentspace[\manifold]{\gamma(t)}$. We
say $U$ is parallel along $\gamma$ if $\nabla_{\gamma^{\prime}(t)}U=0$
for all $t\in I$. In this paper, we primarily focus on parallel vector fields along
geodesics. Let $\gamma:[0,1]\rightarrow\manifold$ be a geodesic connecting
$y$ and $z$, and $U$ a parallel vector field along $\gamma$ such
that $U(0)=u$ and $U(1)=v$. Then $v$ is the parallel transport
of $u$ along $\gamma$, denoted by $\tau_{y,z}u=v$.

When $\manifold$ is a group such that the group operation $\groupop$
and inverse $\iota:y\mapsto y^{-1}$ are smooth, $(\manifold,\groupop)$
is called a Lie group. The tangent space at the identity element $e$ is a Lie algebra denoted by $\mathfrak{g}$. It
consists of left-invariant vector fields $U$, i.e., $U(y\groupop z)=(DL_{y})(U(z))$,
where $L_{y}:z\mapsto y\groupop z$ and $DL_{y}$ is the differential of $L_{y}$.
A Riemannian metric $g$ is called left-invariant if $g_{z}(u,v)=g_{y\groupop z}((DL_{y})u,(DL_{y})v)$
for all $y,z\in\manifold$ and $u,v\in\tangentspace[\manifold]z$,
{i.e., $DL_y$ is an isometry for all $y \in \manifold$}.
Right-invariance can be defined in a similar fashion. A  metric is
bi-invariant if it is both left-invariant and right-invariant. The
Lie exponential map, denoted by $\lieexp$ that maps $\mathfrak{g}$
into $\manifold$, is defined by $\lieexp(u)=\gamma(1)$ where $\gamma:\real\rightarrow\manifold$
is the unique one-parameter subgroup such that $\gamma^{\prime}(0)=u\in\mathfrak{g}$.
Its inverse, if it exists, is denoted by $\lielog$. When $g$ is
bi-invariant, then $\lieexp=\Exp_{e}$, i.e., the Riemannian exponential
map at the identity element coincides with the Lie exponential map. 

\subsection{Additive models for symmetric positive-definite matrices}

The space of $m\times m$ symmetric positive-definite
matrices $\spd$ is  a smooth submanifold of $\real^{m\times m}$,
and its tangent spaces are identified with $\sym$, the collection
of $m\times m$ symmetric matrices. Upon endowing the tangent spaces with a  Riemannian metric
$g$, $\spd$ becomes  a Riemannian manifold.
We specifically focus on the Log-Cholesky \citep{Lin2019a} and Log-Euclidean
\citep{Arsigny2007} metrics while we also consider  extensions to other metrics
and general Riemannian manifolds.  Each
of these  metrics is associated with a group operation $\groupop$
that turns $\spd$ into an abelian Lie group in which the metric is bi-invariant.
\begin{example}[Log-Cholesky metric]
Let $\lowtri$ be the space of $m\times m$
lower triangular matrices and $\chspace\subset\lowtri$ the subspace
such that $L\in\chspace$ if all diagonal elements of $L$ are positive.
One can show that $\chspace$ is a smooth submanifold of $\lowtri$
and its tangent spaces are identified with $\lowtri$. For a fixed
$L\in\chspace$, we define a Riemannian metric $\tilde{g}$ on $\chspace$
by $\tilde{g}_{L}(A,B)=\sum_{1\leq j<i\leq m}A_{ij}B_{ij}+\sum_{j=1}^{m}A_{jj}B_{jj}L_{jj}^{-2}$,
where $A_{ij}$ denotes the element of $A$ in the $i$th row and
$j$th column. It is further turned into an abelian Lie group with
the operation $\varocircle$ defined by $L_{1}\varocircle L_{2}=\mathfrak{L}(L_{1})+\mathfrak{L}(L_{2})+\mathfrak{D}(L_{1})\mathfrak{D}(L_{2})$,
where $\mathfrak{L}(L)$ is the strict lower triangular part of $L$,
that is, $(\mathfrak{L}(L))_{ij}=L_{ij}$ if $j<i$ and $(\mathfrak{L}(L))_{ij}=0$
otherwise, and $\mathfrak{D}(L)$ is the diagonal part of $L$, that
is, a diagonal matrix whose diagonals are equal to the respective diagonals of $L$.
One can show that $\tilde{g}$ is a bi-invariant metric for the Lie
group $\chspace$ with the group operation $\varocircle$. It is well known that a symmetric positive-definite matrix $P$ is
associated with a unique matrix $L$ in $\chspace$ such that $LL^{\top}=P$.
This $L$ is called the Cholesky factor of $P$ in this paper. For
$U,V\in\tangentspace[\spd]P=\sym$, we define the metric $g_{P}(U,V)=\tilde{g}_{L}(L(L^{-1}UL^{-\top})_{\frac{1}{2}},L(L^{-1}VL^{-\top})_{\frac{1}{2}})$,
where $(S)_{\frac{1}{2}}=\mathfrak{L}(S)+\mathfrak{D}(S)/2$ for a
matrix $S$. We also turn $\spd$ into an abelian Lie group with the
operator $\oplus$ such that $P_{1}\oplus P_{2}=(L_{1}\varocircle L_{2})(L_{1}\varocircle L_{2})^{\top}$,
where $L_{1}$ and $L_{2}$ are the Cholesky factors of $P_{1}$ and
$P_{2}$, respectively. The metric $g$ is a bi-invariant metric of
the Lie group $\spd$ with the group operation $\oplus$.
\end{example}

\begin{example}[Log-Euclidean metric]
For a symmetric matrix $S$, % $\exp(S)$ of $S$ is
%defined by 
$\exp(S)=I_{m}+\sum_{j=1}^{\infty}\frac{1}{j!}S^{j}$
 is a symmetric positive-definite matrix. For a symmetric positive-definite
matrix $P$, the matrix logarithmic map is $\log(P)=S$ such that
$\exp(S)=P$. The $\log$ map is a smooth map from the manifold $\spd$ to the space $\sym$. The operation  $\groupop$ defined as  $P_{1}\groupop P_{2}=\exp(\log(P_{1})+\log(P_{2}))$ turns $\spd$ into an abelian group. Define $g_{P}(U,V)=\mathrm{trace}\,\Big[\big((D_{P}\log)U\big)\big((D_{P}\log)V\big)\Big]$,
where $D_{P}\log$ denotes the differential of the  $\log$ map at $P$. 
This is a bi-invariant metric on $\spd$ with the group operation
$\groupop$.
\end{example}

For random elements  $Y\in\spd$ we define the Fr\'echet function
$F(y)=\expect d^{2}(y,Y)$, {where $d$ is the Riemannian distance function induced by
 the Log-Cholesky or the  Log-Euclidean metric}. If $F(y)<\infty$ for some $y\in\spd$
and hence $F(y)<\infty$ for all $y\in\spd$ according to the triangle
inequality, we say $Y$ is of the second order. If $Y$ is a second-order element in $\spd$, then the minimizer of $F(y)$, called the Fr\'echet mean, exists and is unique. This follows from the fact that both
Log-Cholesky and Log-Euclidean metrics turn $\spd$ into a Hamard manifold, i.e., a simply 
connected Riemannian manifold that is also a Hadamard space;  \cite{Sturm2003} showed that  
the Fr\'echet mean exists and is unique for such spaces. 

Given  scalar variables $X_{1}\in\mathcal{X}_{1},\ldots,X_{q}\in\mathcal{X}_{q}$, which are predictors 
that are paired with a manifold-valued response $Y$ and where  $\mathcal{X}_{j} \subset \real ,\, j=1,\dots q,$ are their domains,   we are now in a position to  formulate the proposed manifold additive model as follows,  
\begin{equation}
Y=\mu\groupop w_{1}(X_{1})\groupop\cdots\groupop w_{q}(X_{q})\groupop\zeta,\label{eq:group-model}
\end{equation}
where $\mu$ is the Fr\'echet mean of $Y$, each $w_{k}$ is a function that maps $X_{k}$ into $\manifold$,
 $\zeta$ is random noise which has a  Fr\'echet mean that corresponds to  the group identity
element $e$, and $\mathcal X_k$ are compact domains of $\real$. The above model generalizes the additive model for vector-valued response to $\spd$-valued and more generally Lie group responses. It includes noise impacting the responses, which cannot be  additively modeled  in the absence of a linear structure; the effect of the mean response and of the additive component functions, which again cannot be additively modeled. The Lie group operation is the natural way to substitute for the addition operation in Euclidean spaces when responses lie in a Lie group. 

The statistical task is now to 
 estimate the unknown parameter $\mu$
and the component functions $w_{1},\ldots,w_{q}$, given  a sample  of independently
and identically distributed (i.i.d.) observations of size $n$. This is challenging due to 
the lack of a linear structure in $\spd$ or more generally, for any general Lie group elements.  
The following crucial observation
about the model is the key to tackle this challenge. 
\begin{prop}
\label{prop:Lie-group}If $(\manifold,\groupop)$ is an abelian Lie
group endowed with a bi-invariant metric $g$ that turns $\manifold$ into a Hadamard manifold, then (\ref{eq:group-model}) is equivalent to
\begin{equation}
\Log_{\mu}Y=\sum_{k=1}^{q}\transport_{e,\mu}\lielog w_{k}(X_{k})+\transport_{e,\mu}\lielog\zeta.\label{eq:group-model-Log}
\end{equation}
\end{prop}

Let $f_{k}(X_{k})=\transport_{e,\mu}\lielog w_{k}(X_{k})$ and $\varepsilon=\transport_{e,\mu}\lielog\zeta$.
Then according to Proposition \ref{prop:Lie-group}, one may rewrite the model \eqref{eq:group-model} as
\begin{equation}
\Log_{\mu}Y=\sum_{k=1}^{q}f_{k}(X_{k})+\varepsilon\label{eq:general-model}
\end{equation}
with $\expect\varepsilon=\expect\transport_{e,\mu}\lielog\zeta=\transport_{e,\mu}\expect\lielog\zeta=0$.
{We also note that $\expect\big(\sum_{k=1}^q f_k(X_k)\big)=0$ since $\expect\Log_\mu Y=0$. For the identifiability of the individual component functions $f_k$,
we assume that $\expect f_k(X_k) =0$. This is equivalent to assuming that the Fr\'echet mean of each $w_k(X_k)$ in (\ref{eq:group-model}) equals the group
identity element $e$.}
These considerations  motivate to  estimate the component functions  $w_{k}$ through estimation
of the $f_{k}$, as follows.
\begin{enumerate}[label=\textup{Step \arabic*:},align=left,labelsep=15pt,leftmargin=*]
\item Compute the sample Fr\'echet mean $\hat{\mu}$. Closed-form
expressions of $\hat{\mu}$ are available for many special cases including the Log-Cholesky and Log-Euclidean
metrics.
\item Compute $\Log_{\hat{\mu}}Y_{i}$. There is also a closed-form expression available 
for the Log-Cholesky metric. For the Log-Euclidean metric, there is no closed-form
expression, and a numerical approach is required.
\item Solve the system of integral equations
\begin{equation}\label{sbfeqn}
\hat{f}_{k}(x_{k})=\hat{m}_{k}(x_{k})-{n^{-1}\sum_{i=1}^n \Log_{\hat\mu} Y_i}
-\sum_{j:j\neq k}\int_{\mathcal{X}_{j}}\hat{f}_{j}(x_{j})\frac{\hat{p}_{kj}(x_{k},x_{j})}{\hat{p}_{k}(x_{k})}\diffop x_{j},\quad1\leq k\leq q,
\end{equation}
subject to the constraints $\int_{\mathcal{X}_{j}}\hat{f}_{k}(x_{k})\hat{p}_{k}(x_{k})\diffop x_{k}=0$
for $1\leq k\leq q$. Here, $\hat{p}_{k}(x_{k})=n^{-1}\sum_{i=1}^{n}K_{h_{k}}(x_{k},X_{ik})$,
$\hat{p}_{kj}(x_{k},x_{j})=n^{-1}\sum_{i=1}^{n}K_{h_{k}}(x_{k},X_{ik})K_{h_{j}}(x_{j},X_{ij})$, and
\begin{equation}\label{margest}
\hat m_{k}(x_{k})=n^{-1}\hat{p}_{k}(x_{k})^{-1}\sum_{i=1}^{n}K_{h_{k}}(x_{k},X_{ik})\Log_{\hat{\mu}}Y_{i}.
\end{equation}
Here $K_{h_j}$ is a kernel function with 
$\int_{\mathcal{X}_{j}} K_{h_j}(u,v) \,du=1$ for all $v \in \mathcal{X}_{j}$, see \cite{Jeon2020}. Note that $n^{-1}\sum_{i=1}^n \Log_{\hat\mu} Y_i=0$ since $\hat\mu$ is the sample Fr\'echet mean.

\item Finally, estimate $w_{k}(x_{k})$ by $\hat{w}_{k}(x_{k})=\lieexp\{\tau_{\hat{\mu},e}\hat{f}_{k}(x_{k})\}$.
\end{enumerate}
Step 3 is a multivariate version of the standard
Smooth Backfitting (SBF) system of equations \citep{Mammen1999}. Since the tangent space
$\tangentspace[\spd]{\hat{\mu}}$ is also a Hilbert space, the above
SBF system of equations can be interpreted from a Bochner integral perspective, see  \cite{Jeon2020}, where also the empirical selection of bandwidths $h_k$ is discussed. 

\subsection{Extension to general manifolds}

When $\spd$ is endowed with another metric, such as the affine-invariant
metric \citep{Moakher2005,penn:06,Fletcher2007}, it is no longer
an abelian group with a bi-invariant metric, and Proposition \ref{prop:Lie-group}
does not hold. However, 
model \eqref{eq:general-model} might still
apply, since it depends only on two ingredients, the existence and
uniqueness of the Fr\'echet mean $\mu$, and the well-definedness
of $\Log_{\mu}Y$. These ingredients are satisfied for $\spd$ endowed
with the affine-invariant metric, for which $\spd$ becomes a Hadamard
manifold. For general metrics that might feature positive sectional curvature, or
more generally, for manifolds  beyond $\spd$, we require additional conditions
for  model \eqref{eq:general-model} to be applicable,  as follows.

Let $(\manifold,g)$ now denote a general Riemannian manifold and
$Y$ a random element on $\manifold$. Assume that:\stepcounter{condA}
\begin{enumerate}[label={(A\arabic{condA})},align=left,labelsep=15pt,leftmargin=*]
\item\label{cond:frechet-mean-Y}  The minimizer of $F$ exists and is unique.
\end{enumerate}
As previously mentioned, this condition is satisfied when $\manifold$ is a Hadamard manifold.
For other manifolds, we refer readers to \cite{Bhattacharya2003}
and \cite{afsa:11} for conditions that imply \ref{cond:frechet-mean-Y}. 

For a nonempty subset $A\subset\manifold$,
let $d(y,A)=\inf\{d(y,z):\, z\in A\}$ be the distance between $y$ and the set $A$. For a positive real number $\epsilon$,
we denote $A^{\epsilon}=\{y:\, d(y,A)<\epsilon\}$ and $A^{-\epsilon}=\manifold\backslash (\manifold\backslash A)^\epsilon$. When $A=\emptyset$,
set $A^{\epsilon}=\emptyset$. We make the following
assumption; it  is not needed for the case of a Hadamard manifold. 

\stepcounter{condA}
\begin{enumerate}[label={(A\arabic{condA})},align=left,labelsep=15pt,leftmargin=*]
\item\label{cond:cut-locus}$\prob\{Y\in \mathcal{E}_{\mu}^{-\epsilon_0}\}=1$ for some $\epsilon_0>0$, where $\mathcal{E}_{\mu}$ is defined in Section \ref{subsec:manifold}.
\end{enumerate}
If (A1) and (A2) are satisfied, the proposed manifold additive model \eqref{eq:general-model} remains well
defined, and the first three steps of the estimation method described
in the previous subsection are still valid and can be employed to estimate $f_{1},\ldots,f_{q}$,
with $\spd$ replaced by $\manifold$. %When $\manifold$ is a (potentially
%infinite-dimensional) Hilbert manifold, the estimation method in \citep{Jeon2020}
%can be employed.

\section{Theory\label{sec:Theory}}

We first establish convergence rates and asymptotic normality of the estimators for the mean
and the component functions for general manifolds in the manifold additive model \eqref{eq:general-model} and then provide additional details for  the space $\spd$ endowed with
either the Log-Cholesky metric or the Log-Euclidean metric.
%For a Riemannian manifold $\manifold$, define  $B_{\delta}(\mu)=\{y\in\manifold:d(y,\mu)<\delta\}$
%and  denote by $N(\epsilon,A,d)$ the minimal number of geodesic balls
%of radius $\epsilon$ that are required to cover the subset $A\subset\manifold$
%under the distance function $d$. 
We consider a manifold $\manifold$
that satisfies at least one  of the following two properties but not necessarily both.

\begin{enumerate}[label={(M\arabic*)},align=left,labelsep=15pt,leftmargin=*]

%\item \label{cond:M1}$\manifold$ is a separable Hilbert space.

\item \label{cond:M2}$\manifold$ is a finite-dimensional Hadamard manifold that has sectional curvature
bounded from below by $\mathfrak{c}_{0}\leq0$.

\item \label{cond:M3}\lin{$\manifold$ is a complete compact Riemannian
manifold.} % such that, for all $y\in\manifold$, for some constants $\mathfrak{c}_{3},\rho>0$ that might depend on $y$, for all sufficiently small $\epsilon$, $\log N(\epsilon,B_{\delta}(y),d)\leq\mathfrak{c}_{3}(\delta/\epsilon)^{\rho}$.
\end{enumerate}

The space $\spd$ with the Log-Cholesky metric, Log-Euclidean metric
or affine-invariant metric is a manifold that satisfies  \ref{cond:M2}, \lin{while  
the unit sphere that is used to model
compositional data \citep{dai:17:1} serves as an example of a manifold that satisfies  \ref{cond:M3}.} %The Wasserstein space below serves as an example of  a (non-Riemannian) manifold that satisfies  \ref{cond:M3}. \red{If Riemannian is really needed the W manifold is not covered, otherwise add Riemannian only to (M1)}

%\begin{example}[Wasserstein spaces]
%Let $\mathcal{W}$ be the space of probability distributions
%$G$ on $\real$ such that $\int x^{2}\diffop G<\infty$ and endowed
%with the Wasserstein metric $d_W$ \citep[Example 1]{Petersen2019}. Let
%$(\mathcal{Q},d_{2})$ be the space of quantile functions endowed
%with the $L^{2}$ metric. By Theorem 2.7.5 of \cite{vanderVaart1996},
%$N(\epsilon,\mathcal{W},d_{W})\leq N(\epsilon/2,\mathcal{Q},d_{2})\leq e^{K/\epsilon}$
%for some constant $K>0$. For $Q\in\mathcal{Q}$, let $\mathcal{A}_{\epsilon}(Q)=\{G_{j}:j\in J\}\subset\mathcal{Q}$
%such that $J=N(\epsilon,B_{1}(Q)\cap\mathcal{Q},d_{2})$ and $B_{1}(Q)\subset\bigcup_{j\in J}B_{\epsilon}(G_{j})$.
%Define $\tilde{G}_{j}=Q+\delta(G_{j}-Q)$. Then $B_{\delta}(Q)\subset\bigcup_{j\in J}B_{\epsilon\delta}(\tilde{G}_{j})$.
%This shows that $\log N(\epsilon,B_{\delta}(\omega),d_{W})\leq K\delta/\epsilon$,
%and the entropy condition \ref{cond:M3} is satisfied.
%\end{example}

To establish the convergence rate of $\hat{\mu}$, we also make the
following assumptions.
\begin{enumerate}[label={(A\arabic{condA})},align=left,labelsep=15pt,leftmargin=*]

\stepcounter{condA}
\item\label{cond:manifold}

The manifold $\manifold$ satisfies at least  one of  the conditions %\ref{cond:M1},
\ref{cond:M2} and \ref{cond:M3}.

\begin{comment}
\stepcounter{condA}
\item\label{cond:quadruple}

For some constant $\mathfrak{c}_{1}>0$, $d^{2}(y,u)-d^{2}(z,u)-d^{2}(y,v)+d^{2}(z,v)\leq\mathfrak{c}_{1}d(y,z)d(u,v)$
for all $y,z,u,v\in\manifold$.
\end{comment}

\stepcounter{condA}
\item\label{cond:hessian}

For some constant $\mathfrak{c}_{2}>0$, $F(y)-F(\mu)\geq\mathfrak{c}_{2}d^{2}(y,\mu)$
when $d(y,\mu)$ is sufficiently small.

\stepcounter{condA}
\item\label{cond:hessian-d2} \lin{For some constant $\mathfrak{c}_{3}>0$, for all $y,z\in\manifold$, the linear operator $H_{y,z}:T_z\manifold\rightarrow T_z\manifold$, defined by $g_z(H_{y,z}u,v)=g_z(\nabla_u \Log_zy,v)$ for $u,v\in T_z\manifold$, has its operator norm bounded by $\mathfrak{c}_3\{1+d(z,y)\}$.}
\end{enumerate}
\lin{The operator $H_{y,z}$ in the technical condition \ref{cond:hessian-d2} is indeed the Hessian of the squared distance function $d$; see also the equation (5.4) of \cite{Kendall2011}. It is superfluous if the manifold $\manifold$ is compact, and is satisfied by manifolds of zero curvature. It can also be replaced by a uniform moment condition on the operator norm of $H_{z,Y}$ over all $z$ in a small local neighborhood of $\mu$.} 
We then obtain a parametric convergence rate for the Fr\'echet mean estimates $\hat\mu$.
\begin{prop}
\label{prop:mu-rate}Assume that \ref{cond:frechet-mean-Y}, \ref{cond:manifold} and \ref{cond:hessian} hold and $Y$ is of the second order. Then \newline $d(\hat{\mu},\mu)=\Op(n^{-1/2}).$
\end{prop}

To obtain  convergence rates of the estimated component functions, we
require some additional  conditions that are  standard in the literature
on additive regression.

\begin{enumerate}[label={(B\arabic{condB})},align=left,labelsep=15pt,leftmargin=*]
\stepcounter{condB}
\item\label{cond:kernel}The kernel function $K$ is positive, symmetric,  Lipschitz continuous and supported on $[-1,1]$. 
\stepcounter{condB}
\item \label{cond:bandwidth} The bandwidths $h_k$ satisfy $n^{1/5}h_k\rightarrow \alpha_k>0$.
\stepcounter{condB}
\item \label{cond:X-density} The joint density $p$ of $X_1,\ldots,X_q$ is bounded away from zero
and infinity on $\mathcal X \equiv \mathcal X_1\times\cdots\times\mathcal X_q$.
The densities  $p_{kj}$ are continuously differentiable for $1\leq j\neq k \leq q$.
\stepcounter{condB}
\item\label{cond:smoothness-f} The additive functions $f_k$ are twice continuously (Fr\'echet) differentiable.
\end{enumerate}

Without loss of generality, assume $\mathcal{X}_{k}=[0,1]$ for all
$k$ and let {$\mathcal{I}_{k}=[2h_k,1-2h_k]$.
The moment condition on $\varepsilon$ in the following theorem is required to control the effect of the error of $\hat\mu$ as an estimator of $\mu$
on the discrepancies of $\Log_{\hat{\mu}}Y_{i}$ from $\Log_{\mu}Y_{i}$ after parallel transportation, see Lemma~\ref{error-LogY}.}
It is a mild requirement and is satisfied for example
when the manifold is compact or $\|\varepsilon\|_{\mu}$ follows a sub-exponential
distribution.
\begin{thm}\label{rate}
Assume that the conditions \ref{cond:frechet-mean-Y}--\ref{cond:hessian-d2} and
\ref{cond:kernel}--\ref{cond:smoothness-f} are satisfied, that {$\expect\|\varepsilon\|_{\mu}^\alpha<\infty$ for some
$\alpha \ge 10$ and that $\expect(\|\varepsilon\|_{\mu}^2 \,|\,X_j=\cdot)$ are bounded on $\mathcal X_j$, respectively, for $1 \le j \le q$}.
Then,
\begin{align*}
\max_{1\leq k\leq q}\int_{\mathcal{I}_{k}}\|\tau_{\hat{\mu},\mu}\hat{f}_{k}(x_{k})-f_{k}(x_{k})\|_{\mu}^{2}\,
p_{k}(x_{k})\diffop x_{k}&=\Op(n^{-4/5}),\\
\max_{1\leq k\leq q}\int_{\mathcal{X}_{k}}\|\tau_{\hat{\mu},\mu}\hat{f}_{k}(x_{k})-f_{k}(x_{k})\|_{\mu}^{2}\,
p_{k}(x_{k})\diffop x_{k}&=\Op(n^{-3/5}),
\end{align*}
where $\tau_{\hat{\mu},\mu}$ is the parallel transport operator
along geodesics.
\end{thm}
The following corollary is an immediate consequence. %  of Theorem~\ref{rate}.
\begin{cor}
Under the conditions of Theorem~\ref{rate}, if the $\manifold$ is $\spd$ endowed with either the Log-Cholesky metric or the Log-Euclidean metric, then 
\begin{align*}
\max_{1\leq k\leq q}\int_{\mathcal{I}_{k}}\|\lielog\hat{w}_{k}(x_{k})-\lielog w_{k}(x_{k})\|_e^{2}\,
p_{k}(x_{k})\diffop x_{k}&=\Op(n^{-4/5}),\\
\max_{1\leq k\leq q}\int_{\mathcal{X}_{k}}\|\lielog\hat{w}_{k}(x_{k})-\lielog w_{k}(x_{k})\|_e^{2}\,
p_{k}(x_{k})\diffop x_{k}&=\Op(n^{-3/5}).
\end{align*}
\end{cor}

To derive the asymptotic distribution of $\hat{f}_k$, we define $\mathscr{C}_{k}(x)=\expect\{\varepsilon\otimes \varepsilon\mid X_{k}=x\}$,
where $u\otimes v: T_\mu\manifold \to T_\mu\manifold$ is a tensor product operator such that $(u\otimes v)z=g_{\mu}(u,z)v$. In addition, define
\begin{align}
\Sigma_{k}(x)&=\alpha_{k}^{-1}p_{k}(x)^{-1}\int K(u)^2 \diffop u \cdot \mathscr{C}_{k}(x),\label{eq:Sigma}\\
\delta_{k}(x) &=\frac{p_{k}^{\prime}(x)}{p_{k}(x)}\int u^2 K(u)\,\diffop u \cdot f_{k}^{\prime}(x),\label{eq:delta-k}\\
\delta_{jk}(u,v)&=\frac{\partial p_{jk}(u,v)}{\partial v}\frac{1}{p_{jk}(u,v)}\int u^2 K(u)\,\diffop u \cdot f_{k}^{\prime}(v),\label{eq:delta-jk}\\
\tilde{\Delta}_{k}(x) & =\alpha_{k}^{2}\cdot \delta_{k}(x)+\sum_{j:j\neq k}\alpha_{j}^{2}\int_{\mathcal{X}_{j}}\frac{p_{kj}(x,u)}{p_{k}(x)}
\cdot \delta_{kj}(x,u)\,\diffop u,\label{eq:Delta-k}
\end{align}
where $\alpha_k$ are the constants in the condition \ref{cond:bandwidth}.
Let $(\Delta_{1},\ldots,\Delta_{q})$ be a solution of the system
of equations
\begin{equation}
\Delta_{k}(x)=\tilde{\Delta}_{k}(x)-\sum_{j:j\neq k}\int_{\mathcal{X}_{j}}\frac{p_{kj}(x,u)}{p_{k}(x)}
\cdot \Delta_{j}(u)\diffop u,\quad1\leq k\leq q, \label{eq:Delta-k-eq}
\end{equation}
satisfying the constraints
\begin{equation}
\int_{\mathcal{X}_{k}}p_{k}(x)\cdot \Delta_{k}(x)\diffop x
=\alpha_{k}^{2}\cdot \int_{\mathcal{X}_{k}}p_{k}(x) \cdot \delta_{k}(x)\diffop x,\quad1\leq k\leq q.\label{eq:Delta-k-constraint}
\end{equation}
Finally, define $c_{k}(x)=\frac{1}{2}\int u^2 K(u)\,\diffop u \cdot f_k^{\prime\prime}(x)$
and $\theta_{k}(x)=\alpha_{k}^{2}\cdot c_{k}(x)+\Delta_{k}(x)$.

We assume that
\begin{enumerate}[label={(B\arabic{condB})},align=left,labelsep=15pt,leftmargin=*]
	\stepcounter{condB}
	\item\label{cond:C-continuous}
	
	$\expect\{\varepsilon\otimes \varepsilon\mid X_{k}=\cdot\}$ are continuous operators on $\mathcal X_k$ for all $1 \le k \le q$ and operators 
	$\expect\{\varepsilon\otimes \varepsilon\mid X_j = \cdot, X_{k}=\cdot\}$ are bounded on $\mathcal X_j \times \mathcal{X}_{k}$
	for all $1 \le j \neq k \le q$.
	
	\stepcounter{condB}
	\item\label{cond:p-diff}
	
	$\partial p/\partial x_{k}$, $k=1,\ldots,q$, exist and are bounded
	on $\mathcal{X}=\prod_{k=1}^{q}\mathcal{X}_{k}$.
	
\end{enumerate}
Note that condition \ref{cond:C-continuous} is superfluous if the random noise $\varepsilon$ is independent of the predictors $X_1,\ldots,X_q$.

Let $N_{\mu}(\mathbf{x})$ be the product measure $N(\theta_{1}(x_{1}),\Sigma_{1}(x_{1}))\times\cdots\times N(\theta_{q}(x_{q}),\Sigma_{q}(x_{q}))$
on $(T_\mu\manifold)^{q}$, where $N(\theta,\Sigma)$ denotes a Gaussian
measure on $T_\mu\manifold$ with the mean vector $\theta$ and covariance
operator $\Sigma$. For a set $A$, let $\mathrm{Int}(A)$ denote the interior of $A$.
\begin{thm}\label{asympdist}Assume that conditions \ref{cond:frechet-mean-Y}--\ref{cond:hessian-d2} and 
	\ref{cond:kernel}--\ref{cond:p-diff} hold, that {$\expect\|\varepsilon\|_{\mu}^{\alpha}<\infty$
		for some $\alpha>10$ and that there exists $\alpha'>5/2$ such that $\expect(\|\varepsilon\|_\mu^{\alpha'}\mid X_{k}=\cdot)$
		are bounded on $\mathcal X_k$ for all $1 \le k \le q$.} Then, for $\mathbf{x}=(x_{1},\ldots,x_{q})\in\mathrm{Int}(\mathcal{X})$,
	it holds that $\left(n^{2/5}\big(\tau_{\hat\mu,\mu}\hat{f}_k(x_{k})-f_k(x_{k})\big):1\leq k\leq q \right)\rightarrow N_\mu(\mathbf{x})$
	in distribution. In addition, $n^{2/5}\left(\sum_{k=1}^{q}\tau_{\hat{\mu},\mu}\hat{f}_{k}(x_{k})-\sum_{k=1}^{q} f_{k}(x_{k})\right)$
	converges to $N_\mu(\theta(\mathbf{x}),\Sigma(\mathbf{x}))$, where $\theta(\mathbf{x})=\sum_{k=1}^{q}\theta_{k}(x_{k})$
	and $\Sigma(\mathbf{x})=\Sigma_{1}(x_{q})+\cdots+\Sigma_{q}(x_{q})$.
\end{thm}

When $\manifold$ is $\spd$ equipped  with either the Log-Cholesky metric or the Log-Euclidean metric, the above asymptotic normality can be formulated on the Lie algebra $\liealgebra$. To this end, assume that  $\Sigma_1^{\textup{SPD}},\ldots,\Sigma_q^{\textup{SPD}}$  and $\Delta_1^{\textup{SPD}},\ldots,\Delta_q^{\textup{SPD}}$ are defined by equations \eqref{eq:Sigma}--\eqref{eq:Delta-k-constraint} with $\mathscr C_k(x)$ and $f_k$ replaced by $\expect\{\lielog\zeta \otimes \lielog\zeta \mid X_k=x\}$ and $\psi_k\define \lielog w_k$, respectively. Also, let $c_k^{\textup{SPD}}= \frac{1}{2}\int u^2 K(u)\,\diffop u \cdot \psi_k^{\prime\prime}(x)$
and $\theta_{k}^{\textup{SPD}}(x)=\alpha_{k}^{2}\cdot c_{k}^{\textup{SPD}}(x)+\Delta_{k}^{\textup{SPD}}(x)$, for $k=1,\ldots,q$.
%To this end, we define $\Sigma_k^{\textup{SPD}}(x)=\alpha_{k}^{-1}p_{k}(x)^{-1}\int K(u)^2 \diffop u \cdot \expect\{\lielog\zeta \otimes \lielog\zeta \mid X_k=x\}$. Let $\delta_k^{\textup{SPD}}(x)$ and $\delta_{jk}^{\textup{SPD}}(u,v)$ be defined in analogy to $\delta_k$ and $\delta_{jk}$ by replacing the $f_k$ in \eqref{eq:delta-k} and \eqref{eq:delta-jk} with $\psi_k=\lielog w_k$. Similarly, we define $\tilde{\Delta}_k^{\textup{SPD}}$ in analogy to $\tilde{\Delta}_k$ by replacing $\delta_k$ and $\delta_{jk}$ in \eqref{eq:Delta-k} with $\delta_k^{\textup{SPD}}$ and $\delta_{jk}^{\textup{SPD}}$, respectively. The quantity $\Delta^{\text{SPD}}_k$ is defined in an analogous fashion for each $k=1,\dots,q$. Finally, let $c_k^{\textup{SPD}}= \frac{1}{2}\int u^2 K(u)\,\diffop u \cdot \psi_k^{\prime\prime}(x)$ and $\theta_{k}^{\textup{SPD}}(x)=\alpha_{k}^{2}\cdot c_{k}^{\textup{SPD}}(x)+\Delta_{k}^{\textup{SPD}}(x)$.
The following corollary is an immediate consequence of Theorem~\ref{asympdist}, by noting that the manifold $\spd$ when equipped  with the Log-Cholesky metric or the Log-Euclidean metric satisfies the conditions \ref{cond:frechet-mean-Y}--\ref{cond:hessian} when the second moment of the random noise $\zeta$ is finite.
\begin{cor}
	Assume that the conditions  
	\ref{cond:kernel}--\ref{cond:p-diff} hold and that {$\expect\|\lielog\zeta\|_{\mu}^{\alpha}<\infty$
		for some $\alpha>10$. Furthermore, assume that there exists $\alpha'>5/2$ such that $\expect(\|\lielog\zeta\|_e^{\alpha'}\mid X_{k}=\cdot)$
		are bounded on $\mathcal X_k$ for all $1 \le k \le q$.}
	For $\spd$ endowed with either the Log-Cholesky metric or the Log-Euclidean metric, for $\mathbf{x}=(x_{1},\ldots,x_{q})\in\mathrm{Int}(\mathcal{X})$,
	it holds that $\left(n^{2/5}\big(\lielog\hat{w}_k(x_{k})-\lielog w_k(x_{k})\big):1\leq k\leq q \right)\rightarrow N_{I_m}(\mathbf{x})$
	in distribution. In addition, $n^{2/5}\left(\sum_{k=1}^{q}\lielog\hat{w}_{k}(x_{k})-\sum_{k=1}^{q} \lielog w_{k}(x_{k})\right)$
	converges to $N_{I_m}(\theta(\mathbf{x}),\Sigma(\mathbf{x}))$, where $I_m$ is the $m\times m$ identity matrix,  $\theta(\mathbf{x})=\sum_{k=1}^{q}\theta_{k}^{\textup{SPD}}(x_{k})$
	and $\Sigma(\mathbf{x})=\Sigma_{1}^{\textup{SPD}}(x_{q})+\cdots+\Sigma_{q}^{\textup{SPD}}(x_{q})$.
\end{cor}

\section{Simulations\label{sec:Simulations}}

To illustrate the numerical performance of the proposed manifold additive model estimators, we conducted
simulations for $\manifold=\spd$ for $m=3$ endowed
with the Log-Cholesky metric. We set $\mathcal{X}_{k}=[0,1]$
for $k=1,\ldots,q$. The predictors $X_1,\ldots,X_k$ are independently and identically sampled from the uniform distribution on $[0,1]$. 
%The component functions are $\lielog w_{1}(x)=(x-1/2)S$, $\lielog w_{2}(x)=U\Lambda(x)U^{\top}$ with $\Lambda(x)=\mathrm{diag}(x-1/2,2x^{2}-2/3,3x^{3}-3/4)$, and $\lielog w_{3}(x)=\sin(2\pi x)Q$, where $S,Q$ are some randomly generated symmetric matrices that are fixed throughout all simulation replicates, and similarly, $U$ is a randomly selected unitary matrix that is fixed throughout the simulation study. 
We also fix $\mu$ to be the identity matrix. We then generate the response variable $Y$ by $Y=\mu\groupop w(X_1,\ldots,X_q)\groupop \zeta$, where $w(X_1,\ldots,X_q)=\lieexp\tau_{\mu,e}f(X_1,\ldots,X_q)$ with three settings for $f$:
\begin{itemize}
	\item[I.] $f(x_1,\ldots,x_q)=\sum_{k=1}^q f_k(x_k)$ with $f_k(x_k)$ being an $m\times m$ matrix whose $(j,l)$-entry is $g(x_k;j,l,q)=\exp(-|j-l|/q)\sin(2q\pi(x_k-(j+l)/q))$;
	\item[II.] $f(x_1,\ldots,x_q)=f_{12}(x_1,x_2)+\sum_{k=3}^q f_k$, where $f_k$ is defined as in the setting I, while $f_{12}(x_1,x_2)$ is an $m\times m$ matrix whose $(j,l)$-entry is $g(x_1;j,l,q)g(x_2,j,l,q)$;
	\item[III.] $f(x_1,\ldots,x_q)=f_{12}(x_1,x_2)\prod_{k=3}^q f_k(x_k)$, where $f_{12}(x_1,x_2)$ is an $m\times m$ matrix whose $(j,l)$-entry is $\exp\{-(j+l)(x_1+x_2)\}$, and $f_k(x_k)$ is  an $m\times m$ matrix whose $(j,l)$-entry is $\sin(2\pi x_k)$.
\end{itemize}

The random noise $\zeta$ is generated according to $\lielog \zeta =\sum_{j=1}^6 Z_j v_j$, where $Z_1,\ldots,Z_6$ are independently sampled from $N(0,\sigma^2)$, and $v_1,\ldots,v_6$ are an orthonormal basis of the tangent space $T_e\spd$. The signal-to-ratio (SNR) is measured by $\textup{SNR}=\expect{\|\lielog w(X_1,\ldots,X_q)\|_e^2}/\expect \|\lielog\zeta\|_e^2$. We tweak the value of  the parameter $\sigma^2$ to cover two settings for the  SNR, namely, $\textup{SNR}=2$ and $\textup{SNR}=4$. We note  that the model for $f$  in I is an additive model, while models  II and III are not additive. In particular, 
model III has no additive components and thus represents the most challenging scenario for the proposed additive regression. We consider $q=3$ and $q=4$ to probe the effect of the dimensionality of the predictor vector and sample sizes $n=50,100,200.$

\begin{table}[t]
	\caption{Prediction RMSE and its Monte Carlo standard error\label{tab:pmse}}
	\begin{center}
		\renewcommand*{\arraystretch}{1.2}
		\begin{tabular}{|c|c|c|c|c|c|c|}
			\hline
			\multirow{2}{*}{Setting} & \multirow{2}{*}{$q$} & \multirow{2}{*}{$n$} & \multicolumn{2}{c|}{MAM} & \multicolumn{2}{c|}{ILPR}\tabularnewline
			\cline{4-7}
			& & & SNR=2 & SNR=4 & SNR=2 & SNR=4\tabularnewline
			\hline\multirow{6}{*}{I} & \multirow{3}{*}{3} & 50 & 0.591 (0.056) & 0.508 (0.057) & 1.046 (0.147) & 1.042 (0.146)  \tabularnewline
			\cline{3-7}
			& & 100 & 0.413 (0.026) & 0.339 (0.020) & 0.912 (0.076) & 0.909 (0.092)  \tabularnewline
			\cline{3-7}
			& & 200 & 0.300 (0.017) & 0.230 (0.012) & 0.787 (0.030) & 0.785 (0.050)  \tabularnewline
			\cline{2-7}
			& \multirow{3}{*}{4} & 50 & 0.772 (0.062) & 0.685 (0.063) & 1.075 (0.100) & 1.056 (0.100)  \tabularnewline
			\cline{3-7}
			& & 100 & 0.523 (0.029) & 0.436 (0.036) & 0.964 (0.033) & 0.952 (0.040)  \tabularnewline
			\cline{3-7}
			& & 200 & 0.354 (0.019) & 0.284 (0.013) & 0.918 (0.026) & 0.902 (0.024)  \tabularnewline
			\hline
			\multirow{6}{*}{II} & \multirow{3}{*}{3} & 50 & 0.624 (0.029) & 0.581 (0.024) & 0.948 (0.208) & 0.914 (0.208)  \tabularnewline
			\cline{3-7}
			& & 100 & 0.544 (0.017) & 0.516 (0.013) & 0.769 (0.078) & 0.755 (0.195)  \tabularnewline
			\cline{3-7}
			& & 200 & 0.498 (0.009) & 0.481 (0.008) & 0.645 (0.048) & 0.624 (0.115)  \tabularnewline
			\cline{2-7}
			& \multirow{3}{*}{4} & 50 & 0.687 (0.035) & 0.619 (0.032) & 1.069 (0.150) & 1.054 (0.158)  \tabularnewline
			\cline{3-7}
			& & 100 & 0.553 (0.023) & 0.503 (0.018) & 0.933 (0.088) & 0.924 (0.095)  \tabularnewline
			\cline{3-7}
			& & 200 & 0.471 (0.014) & 0.438 (0.010) & 0.862 (0.045) & 0.838 (0.040)  \tabularnewline
			\hline
			\multirow{6}{*}{III} & \multirow{3}{*}{3} & 50 & 0.801 (0.067) & 0.789 (0.065) & 0.808 (0.220) & 0.791 (0.269)  \tabularnewline
			\cline{3-7}
			& & 100 & 0.750 (0.045) & 0.744 (0.045) & 0.681 (0.210) & 0.688 (0.258)  \tabularnewline
			\cline{3-7}
			& & 200 & 0.725 (0.050) & 0.721 (0.050) & 0.489 (0.083) & 0.467 (0.138)  \tabularnewline
			\cline{2-7}
			& \multirow{3}{*}{4} & 50 & 0.871 (0.079) & 0.866 (0.079) & 1.000 (0.237) & 1.009 (0.272)  \tabularnewline
			\cline{3-7}
			& & 100 & 0.871 (0.077) & 0.870 (0.078) & 0.874 (0.191) & 0.891 (0.230)  \tabularnewline
			\cline{3-7}
			& & 200 & 0.857 (0.063) & 0.857 (0.064) & 0.776 (0.115) & 0.776 (0.139)  \tabularnewline
			\hline
		\end{tabular}
	\end{center}
\end{table}

The quality of the estimation is measured by the prediction root mean squared error on an independent test dataset of 1000 observations, defined by 
$$\textup{RMSE}=\sqrt{\frac{1}{1000}\sum_{i=1}^{1000} d^2(\hat\mu\groupop\hat{w}_1(\tilde x_{i1})\groupop\cdots\groupop \hat w_q(\tilde{x}_{iq}),\tilde{Y}_i)},$$
where $(\tilde{x}_{i1},\ldots,\tilde{x}_{iq},\tilde Y_i)$, $i=1,\ldots,1000$, are i.i.d. observations in the test data. As a comparison method for the proposed  manifold additive model (MAM), we also implement the intrinsic local polynomial regression (ILPR) proposed in  \cite{Yuan2012}, which is a fully nonparametric approach. Each simulation setting is repeated 100 times, and the Monte Carlo prediction RMSE and its standard error are shown in Table \ref{tab:pmse}.

These results lead to the following observations. 
First, as $q=3$ is increased to $q=4$, the prediction RMSE of both methods increases for most cases, with the increase of ILPR much more prominent in almost all cases. This suggests that MAM is less subject to the curse of dimensionality. Second, when the model is correctly specified as in Setting I, the proposed model outperforms ILPR by a significant margin. When the underlying model is not a fully additive model but contains some additive components, such as the model in Setting II, the MAM approach  still clearly outperforms ILPR. When the true model has no additive components, the fully nonparametric approach ILPR is  favored in some cases, especially when $q=3$. However, in the case $q=4$ and the sample size is relatively small, i.e., $n=50$ or $n=100$, the additive model still enjoys better performance even if misspecified.  In summary, when there are several predictors or the sample size is relatively small, the additive model is often preferrable, and when the number of predictors is limited or the sample size is large, a fully nonparametric approach can be competitive.

\section{Application to Diffusion Tensor Imaging\label{sec:Applications}}
We apply the proposed additive model to study diffusion tensors from Alzheimer's Disease Neuroimaging Initiative\footnote{http://adni.loni.usc.edu/} (ADNI). Diffusion tensors are $3\times 3$ symmetric positive-definite matrices that characterize diffusion of water molecules in tissues and convey rich information about brain tissues with important applications in tractography. They are utilized to investigate the integrity of axons and to aid in the diagnosis of brain related diseases. In statistical modeling,  diffusion tensors are typically considered to be random elements in the space $\dtispd \subset \mathcal{S}$, and were studied by \cite{Fillard2005,arsigny2006,Lenglet2006,Pennec2006,zhou:16,Fletcher2007,Dryden2009,Zhu2009,Pennec2020}, among many others. A traditional Euclidean framework for diffusion tensors suffers from significant swelling effects that undesirably inflate the diffusion tensors \citep{Arsigny2007} and impede their interpretation. Consequently,  statistical models have adopted a non-Euclidean approach for diffusion tensor analysis. In the analysis reported below we use the Log-Cholesky metric \citep{Lin2019a} to analyze diffusion tensors; it  is a metric designed to eliminate  the swelling effect.

We focus on the hippocampus that plays a central role in Alzheimer's disease \citep{Lindberg2012}. In the ADNI study, subjects were invited to visit a center for acquisition of their brain images as well as assessment of their memory, executive functioning  and language ability. For each raw diffusion tensor image, a standard preprocessing protocol that includes  denoising, eddy current and motion correction, skull stripping, bias correction and normalization is applied. 
Diffusion tensors for each hippocampal voxel are derived from the  preprocessed images. Then the Log-Cholesky mean \citep{Lin2019a} of the diffusion tensors is computed. This results in an average diffusion tensor for each raw image. The goal is to study the relation between the average hippocampal diffusion tensor and memory, executive functioning  and language ability of the subject. To this end, we utilize the neuropsychological summary scores available from ADNI and documented in \cite{Gibbons2012}. In this study we only consider visits that feature both a properly acquired diffusion tensor image and neuropsychological summary scores. After excluding visits with outliers and missing values, there are  $590$ data tuples of the form $(Y,X_{1},X_{2},X_{3})$, where $Y$ is the average diffusion tensor, which serves as response, while $X_{1},X_2,X_3$, standardized to the interval $[0,1]$, are the predictors and consist of scores for memory, executive functioning and language ability, respectively;  $181$ are  from cognitively normal (CN) subjects and the remainder from patients  who were diagnosed as having either early mild cognitive impairment, mild cognitive impairment, late mild cognitive impairment or Alzheimer's disease (AD). We applied the proposed manifold additive model \eqref{eq:group-model} to the CN and AD groups, respectively.

The resulting component functions $w_1(x_1),w_2(x_2),w_3(x_3)$ are depicted in Figure \ref{fig:dti}, where each diffusion tensor is visualized as an ellipsoid whose volume corresponds to the determinant of the tensor, and the color encodes fractional anisotropy which describes the degree of anisotropy of diffusion of water molecules. For a $3\times 3$ symmetric positive-definite matrix $A$ that represents a diffusion tensor, its fractional anisotropy is defined by $$\text{FA}=\sqrt{\frac{3}{2}\frac{(\rho_1-\bar \rho)^2+(\rho_2-\bar\rho)^2+(\rho_3-\bar\rho)^2}{\rho_1^2+\rho_2^2+\rho_3^2}}$$ with eigenvalues $\rho_1,\rho_2,\rho_3$  of $A$ and $\bar\rho=(\rho_1+\rho_2+\rho_3)/3$. Larger values of fractional anisotropy suggest that movement of the water molecules is constrained by structures such as white matter fibers. In Figure \ref{fig:dti} the first component function $w_1$ suggests that the diffusion tensors are differently associated with memory for the CN and AD groups. In addition, the function $w_1$ of the CN group overall exhibits larger fractional anisotropy. Similar results are observed for the  associations with language ability. In contrast, the association patterns in the two groups are rather similar for executive functioning. The relatively weak association between the average hippocampal diffusion tensor and executive functioning suggests that the hippocampus may play less of a role for executive functioning. In contrast, the significant memory loss and language impairment that accompany Alzheimer's disease appear to be at least partially mediated by the hippocampus. This is in line with previous findings that the   integrity of the hippocampus is not only critical to memory \citep{muller2005}  but also  important for the  flexible use and processing of language \citep{Duff2012}. 

\begin{figure}[t]  
	\centering
	\begin{minipage}[c]{\linewidth}
		\centering
		\begin{tikzpicture}[scale=0.94, every node/.style={scale=0.94}]
		\newcommand\x{0.8}
		
		\node at (0+\x,0) {\includegraphics[width=0.9\textwidth]{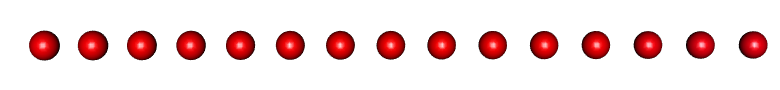}};
		\node at (0+\x,-1.3) {\includegraphics[width=0.9\textwidth]{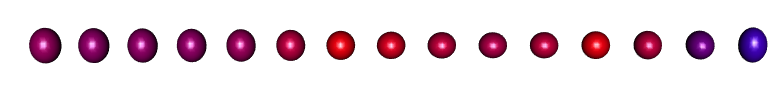}};
		
		\node at (0+\x,-3.3) {\includegraphics[width=0.9\textwidth]{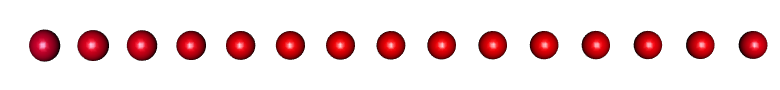}};
		\node at (0+\x,-4.6) {\includegraphics[width=0.9\textwidth]{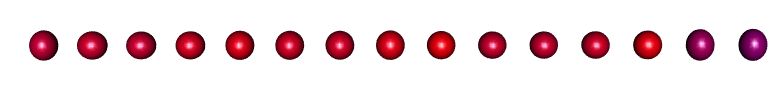}};
		
		\node at (0+\x,-6.6) {\includegraphics[width=0.9\textwidth]{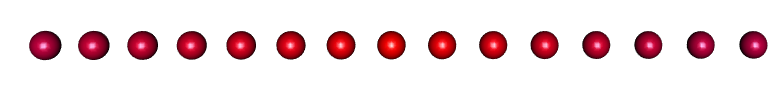}};
		\node at (0+\x,-7.9) {\includegraphics[width=0.9\textwidth]{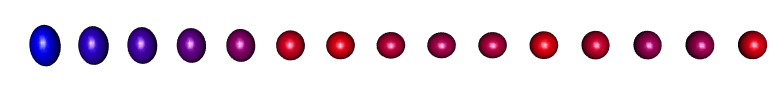}};
		
		\node at (-7.9,-4.2) {\includegraphics[scale=0.5]{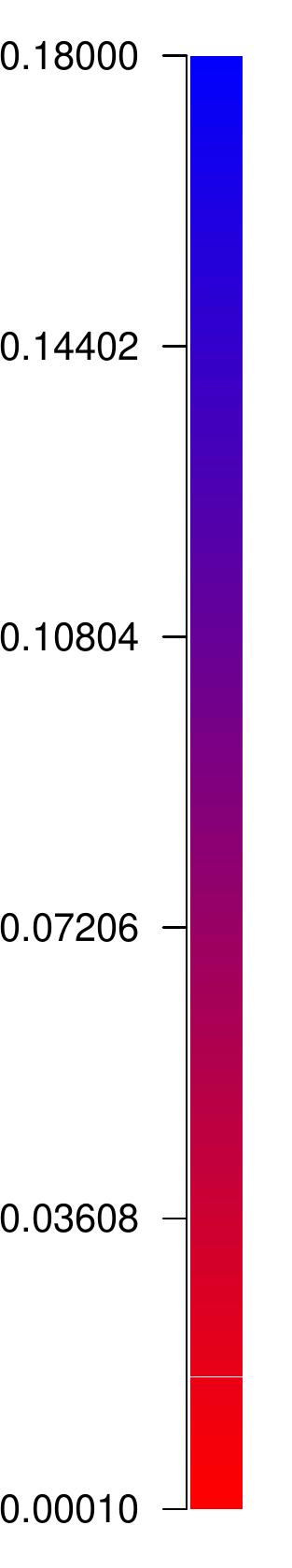}};
		
		\draw[-latex] (-7+\x,-8.8) -- (7.5+\x,-8.8);
		%\node at (-8,-8.85) {age =};
		\draw (-6.7+\x,-8.6) -- (-6.7+\x,-8.8);
		\node at (-6.7+\x,-9.1) {0};
		
		\draw (-4.8+\x,-8.6) -- (-4.8+\x,-8.8);
		\node at (-4.8+\x,-9.1) {0.143};
		
		\draw (-2.9+\x,-8.6) -- (-2.9+\x,-8.8);
		\node at (-2.9+\x,-9.1) {0.286};
		
		\draw (-1.0+\x,-8.6) -- (-1.0+\x,-8.8);
		\node at (-1.0+\x,-9.1) {0.430};
		
		\draw (0.9+\x,-8.6) -- (0.9+\x,-8.8);
		\node at (0.9+\x,-9.1) {0.571};
		
		\draw (2.9+\x,-8.6) -- (2.9+\x,-8.8);
		\node at (2.9+\x,-9.1) {0.714};
		
		\draw (4.9+\x,-8.6) -- (4.9+\x,-8.8);
		\node at (4.9+\x,-9.1) {0.857};
		
		\draw (6.8+\x,-8.6) -- (6.8+\x,-8.8);
		\node at (6.8+\x,-9.1) {1};
		
		\node at (0+\x,-9.7) {Standardized score};
		
		\end{tikzpicture} 
	\end{minipage}
	\caption{Regression of $3 \times 3$ diffusion tensor on memory ($X_1$), executive functioning ($X_2$) and language ability ($X_3$), depicting the estimated additive component functions $w_k$ in model \eqref{eq:group-model}: Component function $w_1(X_1)$ for the  AD group (Row 1) and the CN group (Row 2); function  $w_2(X_2)$ for the  AD group (Row 3) and the CN group (Row 4); function  $w_3(X_3)$ for the  AD group (Row 5) and the CN group (Row 6).  The color encodes the level of fractional anisotropy.\label{fig:dti}}
\end{figure}
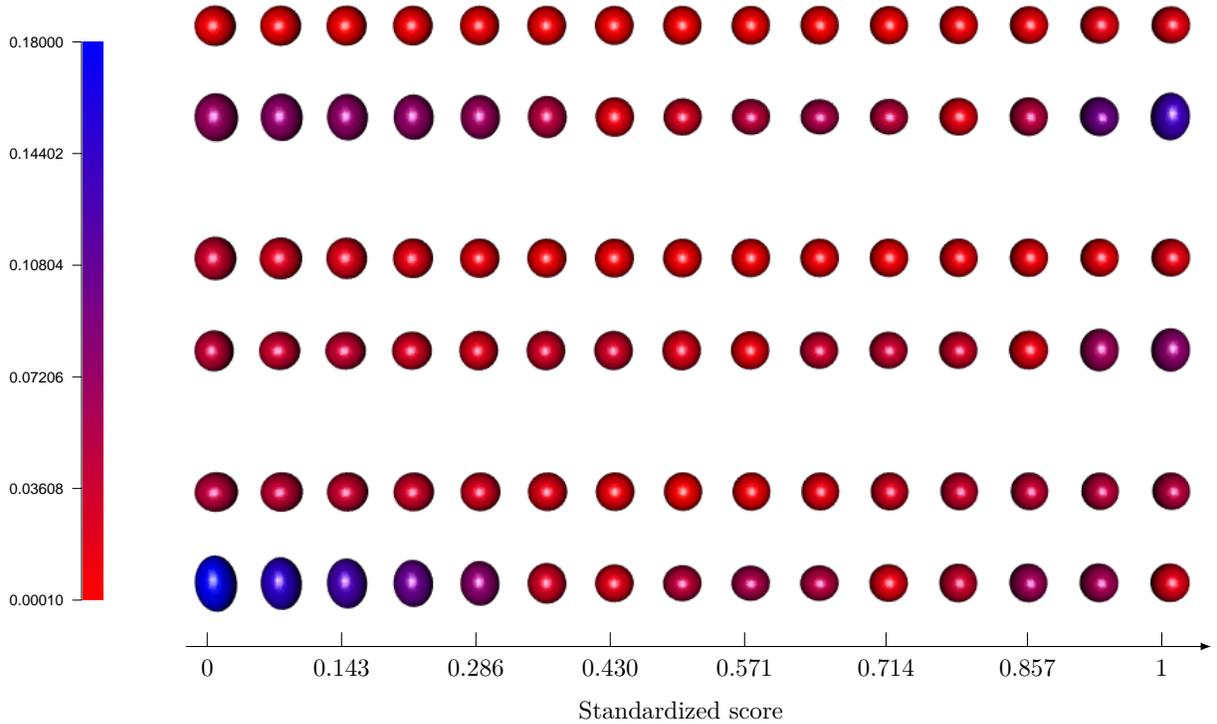

\section*{Appendix: Proofs}

\begin{proof}[Proof of Proposition \ref{prop:Lie-group}]
	First, since $\manifold$ is a Hadamard manifold, %the Fr\'echet mean of $Y$ uniquely exists. In addition, 
	the Riemannian logarithmic map $\Log_\mu$ and  the Lie logarithmic map $\lielog$ are well defined for all elements
	of $\manifold$. Moreover, for a bi-invariant Lie group, the Riemannian
	exponential map $\Exp_{e}$ at the identity element $e$ coincides
	with the Lie exponential map $\lieexp$.
	
	For $z\in\manifold$, let $\mathfrak{u}=\lielog(z)\in\liealgebra$
	and denote by $U$ the associated left-invariant vector field. Define $\gamma_{\mu}(t)=L_{\mu}(\lieexp(t\mathfrak{u})).$
	Then $\gamma_{\mu}^{\prime}(t)=U(\mu\groupop\gamma_{e}(t))$ based
	on the proof of Lemma 6 in \citet{Lin2019a}. The fact that $\gamma_{e}(0)=e$
	further leads to $\gamma_{\mu}^{\prime}(0)=U(\mu)=\tau_{e,\mu}\mathfrak{u}$,
	where the second equality is due to the fact that the parallel transport
	of $\mathfrak{u}$ is realized by the vector field $U$. Noting that
	$\gamma_{\mu}(0)=\mu$, by the definition of the Riemannian exponential
	map,  $\Exp_{\mu}\gamma_{\mu}^{\prime}(0)=\gamma_{\mu}(1)$,
	which leads to $\Exp_{\mu}\tau_{e,\mu}\mathfrak{u}=\ell_{\mu}(\lieexp(\mathfrak{u}))=\mu\groupop z$.
	The last equation is equivalent to $\Log_{\mu}(\mu\groupop z)=\tau_{e,\mu}\lielog(z)$.
	
	Applying the above with $z=w_{1}(X_{1})\groupop\cdots\groupop w_{q}(X_{q})\groupop\zeta$,
	we have $\Log_{\mu}Y=\tau_{e,\mu}\lielog(w_{1}(X_{1})\groupop\cdots\groupop w_{q}(X_{q})\groupop\zeta)=\sum_{k=1}^{q}\tau_{e,\mu}\lielog w_{k}(X_{k})+\tau_{e,\mu}\lielog\zeta$,
	where the second equality stems from $\lieexp(\mathfrak{u}+\mathfrak{v})=\lieexp(\mathfrak{u})\groupop\lieexp(\mathfrak{v})$
	for $\mathfrak{u},\mathfrak{v}\in\liealgebra$ and this  leads to $\lielog(u\groupop v)=\lielog(u)+\lielog(v)$
	for $u,v\in\manifold$.
\end{proof}

\begin{proof}[Proof of Proposition \ref{prop:mu-rate}]
	%When $\manifold$ is a manifold of \ref{cond:M1}, then the conclusion follows from central limit theorem on Hilbert spaces \citep{aitc:86}.
	%If $\manifold$ is a manifold of \ref{cond:M2} or \ref{cond:M3},
	We utilize Corollary 1 of \citet{scho:19}. We first observe that
	$d(\hat{\mu},\mu)=\op(1)$ according to Theorem 2.3 of \citet{Bhattacharya2003}
	and condition \ref{cond:frechet-mean-Y}. Then according to \citet{scho:19}
	the growth and entropy conditions are required to hold only in a neighborhood
	of $\mu$, where the corresponding existence and growth conditions 
	are in assumptions  \ref{cond:frechet-mean-Y} and  \ref{cond:hessian},
	respectively.
	Condition  \ref{cond:frechet-mean-Y} implies that $F$ is finite
	for some point and thus by the triangle inequality for all points  in
	the manifold. If $Z$ is an independent copy of $Y$, then $\expect d^{2}(Y,Z)\leq2\expect\{d^{2}(Y,\mu)\}+2\expect\{d^{2}(\mu,Z)\}=4F(\mu)<\infty$,
	and the moment condition of \citet{scho:19} follows, as well as the 
	weak quadruple condition, where the latter holds for all 
	Hadamard spaces and bounded spaces. %; see Section 3 of \citet{scho:19}
	%for details. 
	Since a compact manifold is a bounded space,
	the weak quadruple condition holds for manifolds under 
	\ref{cond:M2} or  \ref{cond:M3}. 
	
	Finally we verify the entropy condition of \citet{scho:19}. If $\manifold$ is compact,
	then its sectional curvature is bounded away from $-\infty$ and $+\infty$.
	According to the Bishop--G\"unther inequality \citep[][Eq. (3.34)]{Gray2004},
	$\mathrm{vol}(B_\epsilon(z))\geq C_{1}\epsilon^{D}$ for all sufficiently
	small $\epsilon>0$, where $B_\epsilon(z)=\{y\in\manifold:d(y,z)<\epsilon\}$
	and $C_{1}$ is a constant independent of $z$ and $\epsilon$. With
	this result and equation (3.33) of \citet{Gray2004}, the packing
	number and thus the covering number of $B_{\delta}(\mu)$ is bounded
	by $O(\delta^{D}\epsilon^{-D})$. Therefore, the entropy condition
	%of \citet{scho:19} 
	holds for $\alpha=\beta$ for a sufficiently
	small neighborhood of $\mu$ and the result follows. 
	% Using the fact that $F$ has vanishing gradient at $\mu$ and its Hessian at $\mu$ is positive
\end{proof}

The following lemmas are instrumental to establish Theorems \ref{rate} and \ref{asympdist}.

\begin{lem}
\label{lem:max-stat}If $Z_{1},\ldots,Z_{n}$ are nonnegative i.i.d.
random variables with $\expect Z{}_{1}^{\alpha}<\infty$ for some
$\alpha>0$, then $\max_{1\leq i\leq n}Z_{i}=\Op(n^{1/\alpha}).$
\end{lem}
\begin{proof}[Proof of Lemma \ref{lem:max-stat}]
Let $a_{n}=n^{1/\alpha}$. By i.i.d. assumption, for $\epsilon>0$,
\begin{align*}
\prob\left\{ \max_{1\leq i\leq n}Z_{i}\leq Ca_{n}\right\}  & =\left(\prob\{Z_{1}\leq Ca_{n}\right)^{n}=\left(1-\prob\{Z_{1}>Ca_{n}\}\right)^{n}\\
 & \geq\left(1-\frac{\expect Z_{1}^{\alpha}}{C^{\alpha}a_{n}^{\alpha}}\right)^{n}=\left(1-\frac{\expect Z_{1}^{\alpha}}{C^{\alpha}n}\right)^{n}\rightarrow e^{-\expect Z_{1}^{\alpha}/C^{\alpha}}\\
 & \geq1-\epsilon
\end{align*}
for a sufficiently large $C$ that depends on $\epsilon>0$.
\end{proof}

\begin{lem}\label{error-LogY}
Assume the conditions \ref{cond:frechet-mean-Y}--\ref{cond:hessian-d2} and \ref{cond:smoothness-f}. If {$\expect\|\varepsilon\|_{\mu}^{\alpha}<\infty$
for some $\alpha>2$}, then
\[
\max_{1\leq i\leq n}\|\tau_{\hat{\mu},\mu}\Log_{\hat{\mu}}Y_{i}-\Log_{\mu}Y_{i}\|_{\mu}=\Op(n^{-(\alpha-2)/(2\alpha)}).
\]
\end{lem}
\begin{proof}[Proof of Lemma \ref{error-LogY}]
%Using a Taylor expansion similar to equation (16) of \citet{Lin2019} or Theorem 3 of \citet{Pennec2019},  coupled with the mean value theorem and the bound on the sectional curvature of the manifold, 
Using the inequality (5.7) of \cite{Kendall2011}, the condition \ref{cond:hessian-d2} and $d(\hat\mu,\mu)=\op(1)$ that is guaranteed by Proposition \ref{prop:mu-rate},
we deduce that,  with probability tending to one, 
\begin{align*}
\max_{1\leq i\leq n}\|\tau_{\hat{\mu},\mu}\Log_{\hat{\mu}}Y_{i}-\Log_{\mu}Y_{i}\|_{\mu} & =O(1)d(\hat{\mu},\mu)\max_{1\leq i\leq n}\|\Log_{\mu}Y_{i}\|_{\mu}.
\end{align*}
By Lemma \ref{lem:max-stat}, the moment condition $\expect\|\varepsilon\|_{\mu}^{\alpha}<\infty$, the compactness of $\mathcal X$ and the continuity of $f_1,\ldots,f_q$ assumed in \ref{cond:smoothness-f}, 
we have $\max_{1\leq i\leq n}\|\Log_{\mu}Y_{i}\|_{\mu}=\Op(n^{1/\alpha})$.
The conclusion of the lemma then follows from Proposition \ref{prop:mu-rate}.
\end{proof}

\begin{proof}[Proof of Theorem \ref{rate}] We sketch the proof.
Define $\tilde m_j$ as $\hat m_j$ in (\ref{margest}) with $\Log_{\hat\mu} Y_i$ being
replaced by $\Log_\mu Y_i$. Let $(\tilde f_j: 1 \le j \le q)$ denote the solution of the system of equations
\[
\tilde{f}_{k}(x_{k})=\tilde{m}_{k}(x_{k})-n^{-1}\sum_{i=1}^n \Log_{\mu} Y_i
-\sum_{j:j\neq k}\int_{\mathcal{X}_{j}}\tilde{f}_{j}(x_{j})\frac{\hat{p}_{kj}(x_{k},x_{j})}{\hat{p}_{k}(x_{k})}\diffop x_{j},\quad 1\leq k\leq q,
\]
subject to the constraints $\int_{\mathcal X_k}\tilde f_k(x_k)\hat p_k(x_k)\diffop x_{k}=0$ for $1 \le k \le q$.
According to the theory of \cite{Jeon2020}, under the conditions of the theorem, the solution exists and is unique with probability
tending to one. Furthermore, it holds that
\begin{eqnarray}\label{rate-mu-based}\begin{split}
\max_{1\leq k\leq q}\int_{\mathcal{I}_{k}}\|\tilde{f}_{k}(x_{k})-f_{k}(x_{k})\|_{\mu}^{2}\,
p_{k}(x_{k})\diffop x_{k}&=\Op(n^{-4/5}),\\
\max_{1\leq k\leq q}\int_{\mathcal{X}_{k}}\|\tilde{f}_{k}(x_{k})-f_{k}(x_{k})\|_{\mu}^{2}\,
p_{k}(x_{k})\diffop x_{k}&=\Op(n^{-3/5}).
\end{split}
\end{eqnarray}
Since the smooth backfitting operation at (\ref{sbfeqn}) is linear in response variables and the parallel transport
$\tau_{\hat\mu,\mu}: T_{\hat\mu}(\manifold) \to T_\mu(\manifold)$ is also a linear map, we get that
$(\tau_{\hat\mu,\mu} \hat f_j: 1 \le j \le q)$ is nothing else than the smooth backfitting estimator that one gets from the smooth
backfitting operation with $\tau_{\hat\mu,\mu}\Log_{\hat\mu}Y_i$ as responses. We claim
\begin{equation}\label{rate-pf1}
\max_{1\leq k\leq q}\int_{\mathcal{X}_{k}}\|\tau_{\hat{\mu},\mu}\hat{f}_{k}(x_{k})-\tilde f_{k}(x_{k})\|_{\mu}^{2}\,
p_{k}(x_{k})\diffop x_{k}=\Op(n^{-(\alpha-2)/\alpha}).
\end{equation}
The results (\ref{rate-mu-based}) and (\ref{rate-pf1}) give the theorem.

To prove the claim (\ref{rate-pf1}), let $\delta_k=\tau_{\hat{\mu},\mu}\hat{f}_{k}-\tilde f_{k}$. Then,
$(\delta_j: 1 \le j \le q)$ is the solution of the system of equations
\begin{eqnarray*}\begin{split}
\delta_{k}(x_{k})&=\tau_{\hat\mu,\mu}\hat{m}_{k}(x_{k})-\tilde m_k(x_k) - n^{-1}\sum_{i=1}^n \left(\tau_{\hat\mu,\mu}\Log_{\hat\mu}Y_i
-\Log_{\mu} Y_i\right)\\
&\hspace{4cm} - \sum_{j:j\neq k}\int_{\mathcal{X}_{j}}\delta_{j}(x_{j})\frac{\hat{p}_{kj}(x_{k},x_{j})}{\hat{p}_{k}(x_{k})}\diffop x_{j},\quad 1\leq k\leq q,
\end{split}
\end{eqnarray*}
subject to the constraints $\int_{\mathcal X_k}\delta_k(x_k)\hat p_k(x_k)\diffop x_{k}=0$ for $1 \le k \le q$.
From Lemma~\ref{error-LogY} it follows that
\[
\Big\|\tau_{\hat\mu,\mu}\hat{m}_{k}(x_{k})-\tilde m_k(x_k) - n^{-1}\sum_{i=1}^n \left(\tau_{\hat\mu,\mu}\Log_{\hat\mu}Y_i
-\Log_{\mu} Y_i\right)\Big\|_\mu^2 = \Op(n^{-(\alpha-2)/\alpha})
\]
uniformly for $x_k \in \mathcal{X}_k$ for all $1 \le k \le q$. Using the arguments in the proof of Theorem~4.1 in
\cite{Jeon2020}, we may then prove
\begin{equation}\label{rate-pf2}
\sup_{x_k \in \mathcal{X}_k} \|\delta_k(x_k)\|_\mu^2 = \Op(n^{-(\alpha-2)/\alpha}).
\end{equation}
This gives (\ref{rate-pf1}). 
\end{proof}

\begin{proof}[Proof of Theorem \ref{asympdist}] %Here, we also give a sketch of the proof.
Let $\tilde m_j^A(x_j) = n^{-1}\hat p_j(x_j)^{-1}\sum_{i=1}^n K_{h_j}(x_j,X_{ij}) \varepsilon_i$.
Then we find
\begin{equation}\label{asymp-pf}
n^{2/5}\left(\tilde f_j(x_j) - f_j(x_j)\right) = n^{2/5} \tilde m_j^A(x_j) + \tau_{e,\mu}\left(\frac{1}{2}\alpha_j^2 u_2 \cdot \psi_j''(x_j)
+\Delta_j(x_j)\right) + \op(1)
\end{equation}
for each $x_j \in \mathrm{Int(\mathcal{X})}$ for all $1 \le j \le q$, where $u_2=\int u^2 K(u)\diffop u$ and $\tilde{f}_j$ is defined in the proof of Theorem \ref{rate}. Here, $W_n =\op(1)$ means
$\lim_{n \to \infty} P(\|W\|_\mu > \epsilon) =0$ for all $\epsilon>0$. The assertion (\ref{asymp-pf}) can be proved along the lines of the proof of Theorem~4.3
in \cite{Jeon2020}. The expansion (\ref{asymp-pf}) together with (\ref{rate-pf2}) entails
\[
n^{2/5}\left(\tau_{\hat{\mu},\mu} \hat f_j(x_j) - f_j(x_j)\right) = n^{-3/5} \hat p_j(x_j)^{-1} \sum_{i=1}^n K_{h_j}(x_j,X_{ij})
\varepsilon_i + \theta_j(x_j) +\op(1), \quad 1 \le j \le q.
\]
Here, we have used $\alpha>10$. %Noting that $\|\varepsilon\|_\mu = \|\lielog\zeta\|_e$ because of the isometry of the parallel transport between $T_\mu \manifold$ and the Lie algebra $\liealgebra$, the conditions on the moments of $\varepsilon$ are translated to those of $Z$.
By identifying $T_\mu\manifold$ and its metric $g_\mu$ with the Hilbert space $\mathbb{H}$ and the associated inner product
$\langle \cdot, \cdot \rangle$ in \cite{Jeon2020}, respectively, and utilizing Theorem~1.1 in \cite{Kundu2000}, we may prove
that the joint distribution of $\big(n^{-3/5} \hat p_j(x_j)^{-1} \sum_{i=1}^n K_{h_j}(x_j,X_{ij})
\varepsilon_i: 1 \le j \le q\big)$ converges to $N(\mathbf{x})$. This completes the proof of the theorem. 
\end{proof}

\subsection*{Acknowledgments} 
Data used in preparation of this article were obtained from the Alzheimer's Disease Neuroimaging Initiative (ADNI) database (\url{adni.loni.usc.edu}). As such, the investigators within the ADNI contributed to the design and implementation of ADNI and/or provided data but did not participate in analysis or writing of this report. A complete listing of ADNI investigators can be found at: \url{http://adni.loni.usc.edu/wp-content/uploads/how_to_apply/ADNI_Acknowledgement_List.pdf}.

\bibliographystyle{apalike}
\bibliography{mam}

\begin{thebibliography}{}

\bibitem[Afsari, 2011]{afsa:11}
Afsari, B. (2011).
\newblock Riemannian ${L}^p$ center of mass: {{Existence}}, uniqueness, and
  convexity.
\newblock {\em Proceedings of the American Mathematical Society},
  139(2):655--673.

\bibitem[Arnaudon et~al., 2013]{Arnaudon2013}
Arnaudon, M., Barbaresco, F., and Yang, L. (2013).
\newblock Riemannian medians and means with applications to radar signal
  processing.
\newblock {\em IEEE Journal of Selected Topics in Signal Processing},
  7(4):595--604.

\bibitem[Arsigny et~al., 2006]{arsigny2006}
Arsigny, V., Fillard, P., Pennec, X., and Ayache, N. (2006).
\newblock Log-{E}uclidean metrics for fast and simple calculus on diffusion
  tensors.
\newblock {\em Magnetic Resonance in Medicine}, 56(2):411--421.

\bibitem[Arsigny et~al., 2007]{Arsigny2007}
Arsigny, V., Fillard, P., Pennec, X., and Ayache, N. (2007).
\newblock Geometric means in a novel vector space structure on symmetric
  positive-definite matrices.
\newblock {\em SIAM Journal of Matrix Analysis and Applications},
  29(1):328--347.

\bibitem[Barmpoutis et~al., 2007]{Barmpoutis2007}
Barmpoutis, A., Vemuri, B.~C., Shepherd, T.~M., and Forder, J.~R. (2007).
\newblock Tensor splines for interpolation and approximation of {DT}-{MRI} with
  applications to segmentation of isolated rat hippocampi.
\newblock {\em IEEE transactions on medical imaging}, 26(11):1537--1546.

\bibitem[Bhattacharya and Patrangenaru, 2003]{Bhattacharya2003}
Bhattacharya, R. and Patrangenaru, V. (2003).
\newblock Large sample theory of intrinsic and extrinsic sample means on
  manifolds. {I}.
\newblock {\em The Annals of Statistics}, 31(1):1--29.

\bibitem[Caseiro et~al., 2012]{Caseiro2012}
Caseiro, R., Henriques, J.~F., Martins, P., and Batista, J. (2012).
\newblock A nonparametric {R}iemannian framework on tensor field with
  application to foreground segmentation.
\newblock {\em Pattern Recognition}, 45(11):3997--4017.

\bibitem[Chau and von Sachs, 2019]{Chau2019a}
Chau, J. and von Sachs, R. (2019).
\newblock Intrinsic wavelet regression for surfaces of {H}ermitian positive
  definite matrices.
\newblock {\em arXiv:1808.08764 [stat]}.
\newblock arXiv: 1808.08764.

\bibitem[Cornea et~al., 2017]{Cornea2017}
Cornea, E., Zhu, H., Kim, P., and Ibrahim, J.~G. (2017).
\newblock Regression models on {R}iemannian symmetric spaces.
\newblock {\em Journal of the Royal Statistical Society: Series B (Statistical
  Methodology)}, 79(2):463--482.

\bibitem[Dai and M{\"u}ller, 2018]{dai:17:1}
Dai, X. and M{\"u}ller, H.-G. (2018).
\newblock Principal component analysis for functional data on {R}iemannian
  manifolds and spheres.
\newblock {\em The Annals of Statistics}, 46:3334--3361.

\bibitem[Davis et~al., 2010]{Davis2010}
Davis, B.~C., Fletcher, P.~T., Bullitt, E., and Joshi, S. (2010).
\newblock Population shape regression from random design data.
\newblock {\em International Journal of Computer Vision}, 90(2):255--266.

\bibitem[Dryden et~al., 2009]{Dryden2009}
Dryden, I.~L., Koloydenko, A., and Zhou, D. (2009).
\newblock Non-{E}uclidean statistics for covariance matrices, with applications
  to diffusion tensor imaging.
\newblock {\em The Annals of Applied Statistics}, 3(3):1102--1123.

\bibitem[Duff and Brown-Schmidt, 2012]{Duff2012}
Duff, M.~C. and Brown-Schmidt, S. (2012).
\newblock The hippocampus and the flexible use and processing of language.
\newblock {\em Frontiers in Human Neuroscience}, 6:9.

\bibitem[Fillard et~al., 2005]{Fillard2005}
Fillard, P., Arsigny, V., Ayache, N., and Pennec, X. (2005).
\newblock A {R}iemannian framework for the processing of tensor-valued images.
\newblock In {\em International Workshop on Deep Structure, Singularities, and
  Computer Vision}, pages 112--123.

\bibitem[Fillard et~al., 2007]{Fillard2007}
Fillard, P., Arsigny, V., Pennec, X., M.Hayashi, K., M.Thompson, P., and
  Ayache, N. (2007).
\newblock Measuring brain variability by extrapolating sparse tensor fields
  measured on sulcal lines.
\newblock {\em NeuroImage}, 34(2):639--650.

\bibitem[Fletcher, 2013]{Fletcher2013}
Fletcher, P.~T. (2013).
\newblock Geodesic regression and the theory of least squares on {R}iemannian
  manifolds.
\newblock {\em International Journal of Computer Vision}, 105(2):171--185.

\bibitem[Fletcher and Joshi, 2007]{Fletcher2007}
Fletcher, T. and Joshi, S. (2007).
\newblock Riemannian geometry for the statistical analysis of diffusion tensor
  data.
\newblock {\em Signal Processing}, 87:250--262.

\bibitem[Friston, 2011]{Friston2011}
Friston, K.~J. (2011).
\newblock Functional and effective connectivity: a review.
\newblock {\em Brain Connectivity}, 1(1):13--36.

\bibitem[Gibbons et~al., 2012]{Gibbons2012}
Gibbons, L.~E., Carle, A.~C., Mackin, R.~S., Harvey, D., Mukherjee, S., Insel,
  P., Curtis, S.~M., Mungas, D., and Crane, P.~K. (2012).
\newblock A composite score for executive functioning, validated in
  {Alzheimer's Disease Neuroimaging Initiative (ADNI)} participants with
  baseline mild cognitive impairment.
\newblock {\em Brain Imaging and Behavior}, 6(4):517--527.

\bibitem[Gray, 2004]{Gray2004}
Gray, A. (2004).
\newblock {\em Tubes}.
\newblock Springer Basel AG, second edition.

\bibitem[Han et~al., 2020]{Han2020}
Han, K., M{\"u}ller, H.-G., and Park, B.~U. (2020).
\newblock Additive functional regression for densities as responses.
\newblock {\em Journal of the American Statistical Association},
  115(530):997--1010.

\bibitem[Han and Park, 2018]{Han2018}
Han, K. and Park, B.~U. (2018).
\newblock Smooth backfitting for errors-in-variables additive models.
\newblock {\em The Annals of Statistics}, 46(5):216--2250.

\bibitem[Hein, 2009]{Hein2009}
Hein, M. (2009).
\newblock Robust nonparametric regression with metric-space valued output.
\newblock In {\em Advances in Neural Information Processing Systems}, pages
  718--726.

\bibitem[Hinkle et~al., 2014]{Hinkle2014}
Hinkle, J., Fletcher, P.~T., and Joshi, S. (2014).
\newblock Intrinsic polynomials for regression on {R}iemannian manifolds.
\newblock {\em Journal of Mathematical Imaging and Vision}, 50(1-2):32--52.

\bibitem[Hua et~al., 2017]{Hua2017}
Hua, X., Cheng, Y., Wang, H., Qin, Y., Li, Y., and Zhang, W. (2017).
\newblock Matrix {CFAR} detectors based on symmetrized {K}ullback-{L}eibler and
  total {K}ullback-{L}eibler divergences.
\newblock {\em Digital Signal Processing}, 69(C):106--116.

\bibitem[Huettel et~al., 2008]{Huettel2008}
Huettel, S.~A., Song, A.~W., and McCarthy, G. (2008).
\newblock {\em Functional Magnetic Resonance Imaging}.
\newblock Sinauer Associates, 2nd edition.

\bibitem[Jeon and Park, 2020]{Jeon2020}
Jeon, J.~M. and Park, B.~U. (2020).
\newblock Additive regression with {H}ilbertian responses.
\newblock {\em The Annals of Statistics}, page to appear.

\bibitem[Jung et~al., 2015]{Jung2015}
Jung, S., Schwartzman, A., and Groisser, D. (2015).
\newblock Scaling-rotation distance and interpolation of symmetric
  positive-definite matrices.
\newblock {\em SIAM Journal on Matrix Analysis and Applications},
  36(3):1180--1201.

\bibitem[Kendall and Le, 2011]{Kendall2011}
Kendall, W.~S. and Le, H. (2011).
\newblock Limit theorems for empirical {Fr\'{e}chet} means of independent and
  non-identically distributed manifold-valued random variables.
\newblock {\em Brazilian Journal of Probability and Statistics},
  25(3):323--352.

\bibitem[Kundu et~al., 2000]{Kundu2000}
Kundu, S., Majumdar, S., and Mukherjee, K. (2000).
\newblock Central limit theorems revisited.
\newblock {\em Statistics and Probability Letters}, 47(3):265--275.

\bibitem[Lang, 1999]{Lang1999}
Lang, S. (1999).
\newblock {\em Fundamentals of Differential Geometry}.
\newblock Springer, New York.

\bibitem[Le~Bihan, 1991]{LeBihan1991}
Le~Bihan, D. (1991).
\newblock Molecular diffusion nuclear magnetic resonance imaging.
\newblock {\em Magnetic Resonance Quarterly}, 7(1):1--30.

\bibitem[Lee et~al., 2010]{Lee2010}
Lee, Y.~K., Mammen, E., and Park, B.~U. (2010).
\newblock Backfitting and smooth backfitting for additive quantile models.
\newblock {\em The Annals of Statistics}, 38(5):2857--2883.

\bibitem[Lee et~al., 2012]{Lee2012}
Lee, Y.~K., Mammen, E., and Park, B.~U. (2012).
\newblock Flexible generalized varying coefficient regression models.
\newblock {\em The Annals of Statistics}, 40(3):1906--1933.

\bibitem[Lenglet et~al., 2006]{Lenglet2006}
Lenglet, C., Rousson, M., Deriche, R., and Faugeras, O. (2006).
\newblock Statistics on the manifold of multivariate normal distributions:
  Theory and application to diffusion tensor {MRI} processing.
\newblock {\em Journal of Mathematical Imaging and Vision}, 25(3):423--444.

\bibitem[Lin, 2019]{Lin2019a}
Lin, Z. (2019).
\newblock {R}iemannian geometry of symmetric positive definite matrices via
  {C}holesky decomposition.
\newblock {\em SIAM Journal on Matrix Analysis and Applications},
  40(4):1353--1370.

\bibitem[Lin and M\"uller, 2019]{Lin2019b}
Lin, Z. and M\"uller, H.-G. (2019).
\newblock Total variation regularized {F}r\'echet regression for metric-space
  valued data.
\newblock {\em arxiv}.

\bibitem[Lindberg et~al., 2012]{Lindberg2012}
Lindberg, O., Walterfang, M., Looi, J.~C., Malykhin, N., \"Ostberg, P.,
  Zandbelt, B., Styner, M., Velakoulis, D., \"Orndahl, E., Cavallin, L., and
  Wahlund, L.-O. (2012).
\newblock Shape analysis of the hippocampus in alzheimer's disease and subtypes
  of frontotemporal lobar degeneration.
\newblock {\em Journal of Alzheimer's Disease}, 30(2):355--365.

\bibitem[Mammen et~al., 1999]{Mammen1999}
Mammen, E., Linton, O., and Nielsen, J. (1999).
\newblock The existence and asymptotic properties of a backfitting projection
  algorithm under weak conditions.
\newblock {\em The Annals of Statistics}, 27(5):1443--1490.

\bibitem[Moakher, 2005]{Moakher2005}
Moakher, M. (2005).
\newblock A differential geometry approach to the geometric mean of symmetric
  positive-definite matrices.
\newblock {\em SIAM Journal on Matrix Analysis and Applications},
  26(3):735--747.

\bibitem[M{\"u}ller et~al., 2005]{muller2005}
M{\"u}ller, M.~J., Greverus, D., Dellani, P.~R., Weibrich, C., Wille, P.~R.,
  Scheurich, A., Stoeter, P., and Fellgiebel, A. (2005).
\newblock Functional implications of hippocampal volume and diffusivity in mild
  cognitive impairment.
\newblock {\em Neuroimage}, 28(4):1033--1042.

\bibitem[Park et~al., 2018]{Park2018}
Park, B.~U., Chen, C.-J., Tao, W., and M\"uller, H.-G. (2018).
\newblock Singular additive models for function to function regression.
\newblock {\em Statistica Sinica}, 28:2497--2520.

\bibitem[Pelletier, 2006]{Pelletier2006}
Pelletier, B. (2006).
\newblock Non-parametric regression estimation on closed {R}iemannian
  manifolds.
\newblock {\em Journal of Nonparametric Statistics}, 18(1):57--67.

\bibitem[Pennec, 2006]{Pennec2006}
Pennec, X. (2006).
\newblock Intrinsic statistics on {R}iemannian manifolds: Basic tools for
  geometric measurements,.
\newblock {\em Journal of Mathematical Imaging and Vision}, 25:127--154.

\bibitem[Pennec, 2020]{Pennec2020}
Pennec, X. (2020).
\newblock Manifold-valued image processing with {SPD} matrices.
\newblock In {\em Riemannian Geometric Statistics in Medical Image Analysis},
  pages 75--134. Elsevier.

\bibitem[Pennec et~al., 2006]{penn:06}
Pennec, X., Fillard, P., and Ayache, N. (2006).
\newblock A {R}iemannian framework for tensor computing.
\newblock {\em International Journal of Computer Vision}, 66(1):41--66.

\bibitem[Petersen et~al., 2019]{pete:19}
Petersen, A., Deoni, S., and M\"uller, H.-G. (2019).
\newblock Fr\'echet estimation of time-varying covariance matrices from sparse
  data, with application to the regional co-evolution of myelination in the
  developing brain.
\newblock {\em The Annals of Applied Statistics}, 13(1):393--419.

\bibitem[Petersen and M{\"u}ller, 2016]{pete:16}
Petersen, A. and M{\"u}ller, H.-G. (2016).
\newblock Functional data analysis for density functions by transformation to a
  {{Hilbert}} space.
\newblock {\em The Annals of Statistics}, 44(1):183--218.

\bibitem[Petersen and M\"{u}ller, 2019]{Petersen2019}
Petersen, A. and M\"{u}ller, H.-G. (2019).
\newblock Fr\'{e}chet regression for random objects with {E}uclidean
  predictors.
\newblock {\em The Annals of Statistics}, 47(2):691--719.

\bibitem[Pigoli et~al., 2014]{pigo:14}
Pigoli, D., Aston, J.~A., Dryden, I.~L., and Secchi, P. (2014).
\newblock Distances and inference for covariance operators.
\newblock {\em Biometrika}, 101:409--422.

\bibitem[Rathi et~al., 2007]{Rathi2007}
Rathi, Y., Tannenbaum, A., and Michailovich, O. (2007).
\newblock Segmenting images on the tensor manifold.
\newblock In {\em Proocedings of Computer Vision and Pattern Recognition}.

\bibitem[Scheipl et~al., 2015]{Scheipl2015}
Scheipl, F., Staicu, A.-M., and Greven, S. (2015).
\newblock Functional additive mixed models.
\newblock {\em Journal of Computational and Graphical Statistics},
  24(2):477--501.

\bibitem[Sch{\"o}tz, 2019]{scho:19}
Sch{\"o}tz, C. (2019).
\newblock Convergence rates for the generalized {F}r\'echet mean via the
  quadruple inequality.
\newblock {\em Electronic Journal of Statistics}, 13:4280--4345.

\bibitem[Shi et~al., 2009]{Shi2009}
Shi, X., Styner, M., Lieberman, J., Ibrahim, J.~G., Lin, W., and Zhu, H.
  (2009).
\newblock Intrinsic regression models for manifold-valued data.
\newblock In {\em Medical Image Computing and Computer-Assisted Intervention -
  MICCAI}, volume~12, pages 192--199.

\bibitem[Steinke et~al., 2010]{Steinke2010a}
Steinke, F., Hein, M., and Sch{\"{o}}lkopf, B. (2010).
\newblock Nonparametric regression between general {R}iemannian manifolds.
\newblock {\em SIAM Journal on Imaging Sciences}, 3(3):527--563.

\bibitem[Stone, 1985]{ston:85:1}
Stone, C.~J. (1985).
\newblock Additive regression and other nonparametric models.
\newblock {\em The Annals of Statistics}, 13:689--705.

\bibitem[Sturm, 2003]{Sturm2003}
Sturm, K.-T. (2003).
\newblock Probability measures on metric spaces of nonpositive curvature.
\newblock In {\em Heat kernels and analysis on manifolds, graphs, and metric
  spaces (Paris, 2002), vol. 338 of Contemporary Mathematics}, pages 357--390.
  American Mathematical Society, Providence, RI.

\bibitem[Yu et~al., 2008]{Yu2008}
Yu, K., Park, B.~U., and Mammen, E. (2008).
\newblock Smooth backfitting in generalized additive models.
\newblock {\em The Annals of Statistics}, 36(1):228--260.

\bibitem[Yuan et~al., 2012]{Yuan2012}
Yuan, Y., Zhu, H., Lin, W., and Marron, J.~S. (2012).
\newblock Local polynomial regression for symmetric positive definite matrices.
\newblock {\em Journal of Royal Statistical Society: Series B (Statistical
  Methodology)}, 74(4):697--719.

\bibitem[Zhou et~al., 2016]{zhou:16}
Zhou, D., Dryden, I.~L., Koloydenko, A.~A., Audenaert, K.~M., and Bai, L.
  (2016).
\newblock Regularisation, interpolation and visualisation of diffusion tensor
  images using non-{E}uclidean statistics.
\newblock {\em Journal of Applied Statistics}, 43(5):943--978.

\bibitem[Zhu et~al., 2009]{Zhu2009}
Zhu, H., Chen, Y., Ibrahim, J.~G., Li, Y., Hall, C., and Lin, W. (2009).
\newblock Intrinsic regression models for positive-definite matrices with
  applications to diffusion tensor imaging.
\newblock {\em Journal of the American Statistical Association},
  104(487):1203--1212.

\end{thebibliography}

\end{document}